\documentclass[]{llncs}
\usepackage{hyperref}
\usepackage{floatflt, wrapfig}
\usepackage{xcolor}
\usepackage{amsmath}
\usepackage{amssymb, amsbsy}
\usepackage{graphicx}
\graphicspath{ {images/} }
\usepackage{epstopdf}
\usepackage{color, import}
\usepackage{sectsty}
\usepackage{enumitem}
\usepackage{lineno}

\usepackage{cite}
\usepackage[edges]{forest}
\usepackage[linesnumbered,vlined,ruled]{algorithm2e}
\SetAlFnt{\scriptsize}
\usepackage{setspace}
\usepackage[list=true]{subcaption}
\usepackage[title]{appendix}
\usepackage{multirow}
\usepackage{array}
\usepackage{tikz}
\usetikzlibrary{positioning,arrows,automata}
\usepackage{courier}

\usepackage[list=true]{subcaption}
\usepackage{geometry}
\newcolumntype{P}[1]{>{\centering\arraybackslash}p{#1}}
\newcolumntype{M}[1]{>{\centering\arraybackslash}m{#1}}

\begin{document}
\title{Arbitrary Pattern Formation by Opaque Fat Robots with Lights\thanks{The first two authors are supported by NBHM, DAE, Govt. of India and CSIR, Govt. of India, respectively.}}
%
%
\author{Kaustav Bose\inst{1}\orcidID{0000-0003-3579-1941} \and
Ranendu Adhikary\inst{1}\orcidID{0000-0002-9473-2645} \and
Manash Kumar Kundu\inst{2}\orcidID{0000-0003-4179-8293} \and
Buddhadeb Sau\inst{1}}
%
%
\institute{Department of Mathematics, Jadavpur University, Kolkata, India \and
Gayeshpur Government Polytechnic, Kalyani, India\\
\email{\{kaustavbose.rs, ranenduadhikary.rs, manashkrkundu.rs, buddhadeb.sau\}@jadavpuruniversity.in}
}
\maketitle              
\begin{abstract}
\textsc{Arbitrary Pattern Formation} is a widely studied problem in autonomous robot systems. The problem asks to design a distributed algorithm that moves a team of autonomous, anonymous and identical mobile robots to form any arbitrary pattern given as input. The majority of the existing literature investigates this problem for robots with unobstructed visibility. In a few recent works, the problem has been studied in the obstructed visibility model, where the view of a robot can be obstructed by the presence of other robots. However, in these works, the robots have been modelled as dimensionless points in the plane. In this paper, we have considered the problem in the more realistic setting where the robots have a physical extent. In particular, the robots are modelled as opaque disks. Furthermore, the robots operate under a fully asynchronous scheduler. They do not have access to any global coordinate system, but agree on the direction and orientation of one coordinate axis. Each robot is equipped with an externally visible light which can assume a constant number of predefined colors. In this setting, we have given a complete characterization of initial configurations from where any arbitrary pattern can be formed by a deterministic distributed algorithm. 

\end{abstract}

\section{Introduction}

 \textsc{Arbitrary Pattern Formation} or $\mathcal{APF}$ is a fundamental coordination problem for distributed multi-robot systems. Given a team of autonomous mobile robots,  the goal is to design a distributed algorithm that guides the robots to form any specific but arbitrary geometric pattern given to the robots as input. \textsc{Arbitrary Pattern Formation} is closely related to the \textsc{Leader Election} problem where a unique robot from the team is to be elected as the leader. In the traditional framework of theoretical studies, the robots are modelled as \emph{autonomous} (there is no central control), \emph{homogeneous} (they execute the same distributed algorithm), \emph{anonymous} (they have no unique identifiers) and \emph{identical} (they are indistinguishable by their appearance) computational entities that can freely move in the plane. Each robot is equipped with sensor capabilities to perceive the positions of other robots. The robots do not have access to any global coordinate system. The robots operate in \textsc{Look-Compute-Move} \emph{(LCM)} cycles:  upon becoming active, a robot takes a snapshot of the positions of the other robots (\textsc{Look}), then computes a destination based on the snapshot (\textsc{Compute}), and then moves towards the destination along a straight line (\textsc{Move}).

 The \textsc{Arbitrary Pattern Formation} problem has been extensively studied in the literature in various settings (See \cite{Yamashita99,Yamashita10,Flocchini08,Dieudonne10,Cicerone17,bose18} and references therein). Until recently, the problem had only been studied for robots with unobstructed visibility. In \cite{VaidyanathanST18}, the problem was first considered in the \emph{opaque robots} or \emph{obstructed visibility} model which assumes that the visibility of a robot can be obstructed by the presence of other robots. This is a more realistic model for robots equipped with camera sensors. They also assumed that the robots are equipped with persistent visible lights that can assume a constant number of predefined colors. This is known as the \emph{luminous robot} model, introduced by Peleg \cite{peleg2005distributed}, where the lights serve both as a medium of weak explicit communication and also as a form of memory. In \cite{VaidyanathanST18}, the robots are first brought to a configuration in which each robot can see all other robots, and then \textsc{Leader Election} is solved by a randomized algorithm. The first fully deterministic solutions for \textsc{Leader Election} and \textsc{Arbitrary Pattern Formation} were given in \cite{BoseKAS19} for robots whose local coordinate systems agree on the direction and orientation of one coordinate axis. However, in both \cite{BoseKAS19,VaidyanathanST18}, the robots were modelled as dimensionless points in the plane. This assumption is obviously unrealistic, as real robots have a physical extent. In this work, we extend the results of \cite{BoseKAS19} to the more realistic \emph{opaque fat robots} model \cite{czyzowicz2009gathering,agathangelou2013distributed}. Furthermore, our algorithm also works for robots with \emph{non-rigid} movements (a robot may stop before it reaches its computed destination), whereas the algorithm of \cite{BoseKAS19} requires robots to have \emph{rigid} movements (a robot reaches its computed destination without any interruption). Also, the total number of moves executed by the robots in our algorithm is asymptotically optimal. The contribution of this paper is summarized in Table \ref{tab}.

\begin{table}
\centering
\begin{tabular}{| M{1.3cm} | M{1.2cm}| M{2.2cm}| M{1.6cm}| M{1.2cm}| M{1.7cm}| M{2cm}|} 
\hline
      & Robots  & Agreement in coordinate system & Scheduler & No. of colors used & Movement  & Total no. of moves\\ 
\hline
\cite{BoseKAS19} &	Point	& One-axis agreement		 & \textsc{Async}     &  6	          & Rigid     & $O(n^2)$\\ 
\hline
This paper & Fat  	& One-axis agreement		 & \textsc{Async}     &  10               & Non-rigid & $\Theta(n)$\\ 
\hline
\end{tabular}
\caption{Comparison of this work with previous ones.}
\label{tab}
\end{table}

\section{Model and Definitions}

In this section, we shall formally describe the robot model and also present the necessary definitions and notations that will be used in the rest of the paper.

\textbf{Robots.} We consider a set of $n \geq 3$ autonomous, anonymous, homogeneous and identical fat robots. Each robot is modelled as a disk of diameter equal to 1 unit. The robots do not have access to any global coordinate system, but their local coordinate systems agree on the direction and orientation of the $X$-axis. They also agree on the unit of length as the diameter of the robots are same and taken as 1 unit. 


\textbf{Lights.} Each robot is equipped with an externally visible light which can assume a constant number of colors. Our algorithm will require in total ten colors, namely \texttt{off}, \texttt{terminal}, \texttt{interior}, \texttt{failed}, \texttt{symmetry}, \texttt{ready}, \texttt{move}, \texttt{switch off}, \texttt{leader} and \texttt{done}. Initially all robots have their lights set to \texttt{off}.

\textbf{Visibility.} The visibility range of a robot is unlimited, but can be obstructed by the presence of other robots. Formally, a point $p$ in the plane is visible to a robot $r_i$ if and only if there exists a point $x_i$ on the boundary of $r_i$ such that the line segment joining $p$ and $x_i$ does not contain any point of any other robot. This implies that a robot $r_i$ can see another robot $r_j$ if and only if there is at least one point on the boundary of $r_j$ that is visible to $r_i$. Also, if $r_i$ can see any portion of the boundary of $r_j$, then we assume that it can determine the position of (the center of) $r_j$.

\textbf{Look-Compute-Move cycles.} The robots, when active, operate according to the so-called \textsc{Look-Compute-Move} \emph{(LCM)} cycle. In each cycle, a previously idle robot wakes up and executes the following steps. In \textsc{Look}, a robot takes the snapshot of the positions of the robots visible to it (represented in its own coordinate system), along with their respective colors. Then in \textsc{Compute}, based on the perceived configuration, the robot performs computations according to a deterministic algorithm to decide a destination point and a color. Finally in \textsc{Move}, it sets its light to the decided color and moves towards the destination point. When a robot transitions from one LCM cycle to the next, all of its local memory (past computations and snapshots) are erased, except for the color of the light.

\textbf{Scheduler.} We assume that the robots are controlled by a fully asynchronous adversarial scheduler. The robots are activated independently and each robot executes its cycles independently. The amount of time spent in \textsc{Look, Compute, Move} and inactive states is finite but unbounded, unpredictable and not same for different robots. As a result, a robot can be seen while moving, and hence, computations can be made based on obsolete information about positions. 


\textbf{Movement.} We assume that the robots have \emph{non-rigid} movements. This means that a robot may stop before it reaches its destination. However, there is a fixed $\delta > 0$ so that each robot traverses at least the distance $\delta$ unless its destination is closer than $\delta$. The value of $\delta$, however, is not known to the robots. The existence of a fixed $\delta$ is necessary, because otherwise, a robot may stop after moving distances $\frac{1}{2}, \frac{1}{4}, \frac{1}{8}, \ldots$ and thus, not allowing any robot to traverse a distance of more than 1.

\textbf{Definitions and Notations.} We shall denote the set of robots by $\mathcal{R} = \{r_1, r_2, \dots, r_n\}$, $n \geq 3$. When we say that a robot is at a point $p$ on the plane, we shall mean that its center is at $p$. For any time $t$, the configuration of the robots at time $t$, denoted by $\mathbb{C}(t)$ or simply $\mathbb{C}$, is a sequence $(p_1(t), p_2(t), \dots, p_n(t))$ of $n$ points on the plane, where $p_i(t)$ is the position of (the center of) the robot $r_i$ at $t$. At any time $t$, $r(t).light$ or simply $r.light$ will denote the color of the light of $r$ at $t$. With respect to the local coordinate system of a robot, positive and negative directions of the $X$-axis will be referred to as \emph{right} and \emph{left} respectively, and the positive and negative directions of the $Y$-axis will be referred to as \emph{up} and \emph{down} respectively. Since the robots agree on the $X$-axis, they agree on horizontal and vertical. They also agree on left and right, but not on up and down. For a robot $r$, $\mathcal{L}_{V}(r)$ and $\mathcal{L}_{H}(r)$ are respectively the vertical and horizontal lines passing through the center of $r$. We denote by $\mathcal{H}_{U}^{O}(r)$ (resp. $\mathcal{H}_{U}^{C}(r)$) and $\mathcal{H}_{B}^{O}(r)$ (resp. $\mathcal{H}_{B}^{C}(r)$) the upper and bottom open (resp. closed) half-planes delimited by $\mathcal{L}_{H}(r)$ respectively. Similarly, $\mathcal{H}_{L}^{O}(r)$ (resp. $\mathcal{H}_{L}^{C}(r)$) and $\mathcal{H}_{R}^{O}(r)$ (resp. $\mathcal{H}_{R}^{C}(r)$) are the left and right open (resp. closed) half-planes delimited by $\mathcal{L}_{V}(r)$ respectively. For a configuration $\mathbb{C}$, a subset of robots that are on the same vertical line will be called a \emph{batch}. Thus, any configuration $\mathbb{C}$ can be partitioned into batches $B_1, \dots, B_k$, ordered from left to right. The vertical line passing through the centers of the robots of a batch will be called the \emph{central axis} of that batch. When we say `the distance between a batch $B_i$ and a robot $r$ (resp. another batch $B_j$)', we shall mean the horizontal distance between the central axis of $B_i$ and the center of $r$  (resp. central axis of $B_j$). A robot $r$ belonging to batch $B_i$ will be called \emph{non-terminal} if it lies between two other robots of  $B_i$, and otherwise it will be called \emph{terminal}. Consider any batch $B_j$ whose central axis is $\mathcal{S}$ and a horizontal line $\mathcal{T}$.  Let $\mathcal{H}_1$ and $\mathcal{H}_2$ be the closed half-planes delimited by $\mathcal{T}$. For each $\mathcal{H}_i, i = 1, 2$, consider the distances of the robots on $\mathcal{S} \cap \mathcal{H}_i$ from $\mathcal{T}$ arranged in increasing order. The string of real numbers thus obtained is denoted by $\lambda_i$. To make the lengths of the strings $\lambda_1$ and $\lambda_2$ equal, null elements $\Phi$ may be appended to the shorter string. Now the two strings are different if and only if the robots of $B_j$ are not in symmetric positions with respect to $\mathcal{T}$. In that case, $\mathcal{H}_i$ will be called the \emph{dominant half with respect to $\mathcal{T}$ and $B_j$} if $\lambda_i$ is the lexicographically smaller sequence (setting $x < \Phi$ for any $x \in \mathbb{R}$). 

%
%

\textbf{Problem Definition.} Consider an initial configuration of $n$ fat opaque robots in the Euclidean plane, all having their lights set to \texttt{off}. Each robot is given as input, a pattern $\mathbb{P}$, which is a list of $n$ distinct elements from $\mathbb{R}^2_{\geq 0} = \{ (a,b) \in \mathbb{R}^2 | a, b \geq 0 \}$. The \textsc{Arbitrary Pattern Formation} requires to design a distributed algorithm that guides the robots to a configuration that is similar to $\mathbb{P}$ with respect to translation, reflection, rotation and uniform scaling. 


\section{The Algorithm}

The main result of the paper is Theorem \ref{thm_main}. The proof of the `only if' part is the same as in case for point robots, proved in \cite{BoseKAS19}. The `if' part will follow from the algorithm presented in this section.

\begin{theorem}\label{thm_main}
 For a set of opaque luminous fat robots with non-rigid movements and having one axis agreement, $\mathcal{APF}$ is deterministically solvable if and only if the initial configuration is not symmetric with respect to a line $\mathcal{K}$ which 1) is parallel to the agreed axis and 2) does not pass through the center of any robot.
 \end{theorem}

For the rest of the paper, we shall assume that the initial configuration $\mathbb{C}(0)$ does not admit the unsolvable symmetry stated in Theorem \ref{thm_main}. Our algorithm works in two \emph{stages}, namely \emph{leader election} and \emph{pattern formation from leader configuration}. The first stage is again divided into two phases, namely \emph{Phase 1} and \emph{Phase 2}. In the first stage, a single robot will be elected as the leader of the swarm. Since the robots do not have access to any global coordinate system, they do not agree on how the given pattern $\mathbb{P}$ would be realized in the plane. With the help of the elected leader, the robots can implicitly agree on a common coordinate system. Once an agreement on a common coordinate system is achieved, the robots will arrange themselves to form the given pattern in the agreed coordinate system in the second stage. Since the robots are oblivious, in each LCM cycle, a robot has to infer the current stage or phase from its local view. This is described in Algorithm \ref{main_algorithm}.

  \begin{minipage}{.9\linewidth}
  \begin{algorithm}[H]
    \setstretch{.1}
    \SetKwInOut{Input}{Input}
    \SetKwInOut{Output}{Output}
    \SetKwProg{Fn}{Function}{}{}
    \SetKwProg{Pr}{Procedure}{}{}

    \Input{The configuration of robots visible to me.}
    
    \Pr{\textsc{ArbitraryPatternFormation()}}{

    \uIf(\tcp*[f]{stage 2}){there is a robot with light set to \texttt{leader}}{
    
	\textsc{PatternFormationFromLeaderConfiguration()}
    
      }
      
      \Else(\tcp*[f]{stage 1}){
    
	\uIf{(the first batch has two robots with light set to \texttt{terminal}) and (the lights of all robots of the second batch are set to same color) and (the distance between the first and second batch is at least $\frac{n+3}{2}$ units)}{\textsc{Phase2()}}
	
	\uElseIf{there is at least one robot with light set to \texttt{failed}, \texttt{symmetry}, \texttt{ready}, \texttt{move} or \texttt{switch off}}{
	
	    \textsc{Phase2()} }
	    
	\Else{\textsc{Phase1()}}

    }
    
    }
\caption{Arbitrary Pattern Formation}
    \label{main_algorithm} 
\end{algorithm}
\end{minipage}

\subsection{Leader Election}

In the leader election stage, a unique robot $r_l$ will elect itself as leader by setting its light to \texttt{leader} (while the lights of all other robots should be set to \texttt{off}). We want the configuration to satisfy some additional properties as well, that will be useful in the second stage of the algorithm. In particular, we want 1) all the non-leader robots to lie inside  $\mathcal{H}_{R}^{O}(r_l) \cap \mathcal{H}$ where $\mathcal{H} \in \{\mathcal{H}_{U}^{O}(r_l), \mathcal{H}_{B}^{O}(r_l)\}$, and 2) the distance of any non-leader robot  from $\mathcal{L}_H(r_l)$ to be at least 2 units. We shall call this a \emph{leader configuration}, and call $r_l$ the \emph{leader}.

\subsection{Phase 1}

Since the robots already have an agreement on left and right, if there is a unique leftmost robot, i.e., the first batch has only one robot, then that robot, say $r$, can identify this from its local view and elect itself as the leader. However, the robot $r$ will not immediately change its light to \texttt{leader} as the additional conditions of a leader configuration might not be yet satisfied. So, it will start executing the procedure \textsc{BecomeLeader()} (See Fig. \ref{Fig: proof0}). Only after these conditions are satisfied, $r$ will change its light to \texttt{leader}. However, there might be more than one leftmost robots in the configuration. In the extreme case, all the robots may lie on the same vertical line, i.e., there may be only one batch. So if there are more than one leftmost robots, the aim of Phase 1 is to move the two terminal robots of the first batch leftwards by the same amount. We also want the distance between (the central axes of) the first batch and second batch in the new configuration to be at least $\frac{n+3}{2}$ units. Therefore, at the end of Phase 1, we shall either have a leader configuration or have at least two batches in the configuration with the first batch having exactly two robots and at least $\frac{n+3}{2}$ units to the left of the second batch. In the second case, the lights of the two robots of the first batch will be set to \texttt{terminal}, lights of all robots of the second batch will be set to either \texttt{interior} or \texttt{off}, and all other robots have lights set to \texttt{off}. Due to space restrictions, we will not describe the algorithm here in much detail. A formal pseudocode description of the algorithm is given in Algorithm \ref{algo:phase1}. Further details can be found in Appendix \ref{appendix_p1}.

  \begin{algorithm}[H]
    \setstretch{1}
    \SetKwInOut{Input}{Input}
    \SetKwInOut{Output}{Output}
    \SetKwProg{Fn}{Function}{}{}
    \SetKwProg{Pr}{Procedure}{}{}

    \Pr{\textsc{Phase1()}}{

    $r \leftarrow$ myself
    
    \uIf{$r.light =$ \texttt{off}}{
    
	\uIf{I am in the first batch and I am the only robot in my batch}{\textsc{BecomeLeader()}}
	
	\uElseIf{I am in the first batch and I am not the only robot in my batch}{
	
	  \uIf{I am terminal}{$r.light \leftarrow$ \texttt{terminal}}
	  \Else{$r.light \leftarrow$ \texttt{interior}}
	  }
	  
	\ElseIf{there is a robot with light \texttt{interior} on $\mathcal{L}_{V}(r)$}{
	
	  \uIf{I am not terminal}{$r.light \leftarrow$ \texttt{interior} \label{code: p1_0}}
	  
	  \ElseIf{(I am terminal) and (there is exactly one robot $r'$ in $\mathcal{H}_{L}^{O}(r)$) and ($r'.light =$ \texttt{terminal})}{
	  
	    $r.light \leftarrow$ \texttt{terminal}
	    
	    Move $\frac{n+3}{2}$ units to the left
	    	    
	    }

	}

	}

      \ElseIf{$r.light =$ \texttt{terminal}}{

	\uIf{there is a robot on $\mathcal{L}_{V}(r)$ with light \texttt{interior}}{Move $\frac{n+3}{2}$ units to the left}

      \uElseIf{there is a robot on $\mathcal{L}_{V}(r)$ with light \texttt{terminal}}{
      
	$d \leftarrow$ my horizontal distance from the leftmost robot in $\mathcal{H}_{R}^{O}(r)$

	\If{$d < \frac{n+3}{2}$}{Move $\frac{n+3}{2}-d$ units to the left}
      
      }
      
      \ElseIf{there is a robot $r'$ in $\mathcal{H}_{L}^{O}(r)$ with light \texttt{terminal} \label{code: p1_1}}{
	
	$d \leftarrow$ my horizontal distance from $r'$
	
	Move $d$ units to the left
	
      }
      
      }

    }
 

\caption{Phase1}
    \label{algo:phase1} 
\end{algorithm}

\subsection{Phase 2}\label{sec: p2}

Assume that at the end of Phase 1, we have $k \geq 2$ batches and exactly two robots $r^1_1$ and $r^1_2$ in the first batch $B_1$ with light \texttt{terminal} that are at least $\frac{n+3}{2}$ units to the left of $B_2$. So now we are in Phase 2. Let $\mathcal{L}$ be the horizontal line passing through the mid-point of the line segment joining $r^1_1$ and $r^1_2$. Let $\mathcal{H}_1$ and $\mathcal{H}_2$ be the two open half-planes delimited by $\mathcal{L}$ such that $r^1_1 \in \mathcal{H}_1$ and $r^1_2 \in \mathcal{H}_2$. Our algorithm will achieve the following. Define $i > 1$ to be the smallest integer such that $B_{i}$ is either (Case 1) asymmetric with respect to $\mathcal{L}$, or (Case 2) symmetric with respect to $\mathcal{L}$, but it has a robot lying on $\mathcal{L}$. In Case 1, a terminal robot from $B_{i-1}$ will become the leader and in Case 2, the robot from $B_{i}$ that lies on $\mathcal{L}$ will become the leader. From left to right, terminal robots of different batches will attempt to elect a leader either by electing itself as the leader or by asking a robot of the next batch to become the leader. In particular, when a batch  $B_{j}$ tries to elect leader, its terminal robots  will check whether the next batch $B_{j+1}$ is asymmetric or symmetric with respect to $\mathcal{L}$. In the first case, the terminal robot of $B_{j}$ lying in the dominant half with respect to $\mathcal{L}$ and $B_{j+1}$ will elect itself as the leader. In the later case, the terminal robots of $B_{j}$, using light, will communicate to the robots of $B_{j+1}$ the fact that $B_{j+1}$ is symmetric with respect to $\mathcal{L}$. If $\mathcal{L}$ passes through the center of a robot of $B_{j+1}$, then that robot will elect itself as the leader. Now there are three issues regarding the implementation of this strategy, which we shall discuss in the following three sections.
 
\begin{enumerate} 
 \item What happens when the robots of $B_{j}$ are unable to see all the robots of $B_{j+1}$?
 
  \item How will the robots ascertain $\mathcal{L}$ from their local view?
 
 \item How will all the conditions of a leader configuration be achieved?
\end{enumerate}

\subsubsection{Coordinated Movement of a Batch}

When the terminal robots of a batch $B_j$ attempt to elect leader, they need to see all the robots of the next batch $B_{j+1}$. But since the robots are fat and opaque, a robot may not be able to see all the robots of the next batch (See Fig. \ref{view}). However, each robot of two consecutive batches will be able to see all robots of the other batch if the two batches are more than 1 unit distance apart. Recall that at the beginning of Phase 2, the robots of $B_{1}$ are at least $\frac{n+3}{2}$ units to the left of the robots of $B_{2}$. Therefore, when the robots of the first batch attempt to elect leader, they are able to see all robots of $B_2$. Now consider the case where the terminal robots of $B_j, j > 1$ are trying to elect leader. Therefore, the first $j-1$ batches must have failed to elect leader. This implies that the first $j$ batches are symmetric with respect to $\mathcal{L}$ and $\mathcal{L}$ does not pass through the center of a robot of the first $j$ batches. After the terminal robots of $B_{j-1}$ fail to elect leader, they will change their lights to \texttt{failed} and ask the next batch $B_j$ to try to elect a leader. Then the robots of $B_{j}$ will move left to position themselves exactly at a distance $1 + \frac{1}{n}$ units from the robots of $B_{j-1}$. It can be shown that (See the proof of Theorem \ref{thm_phase2}.) $B_j$ will have sufficient space to execute the movement and also, their horizontal distance from the robots of $B_{j+1}$ will be at least 2 units after the movement. So, after the movements, the terminal robots of $B_j$ can see all the robots of $B_{j+1}$.

However, since the scheduler is asynchronous and the movements are non-rigid, the robots of $B_j$ can start moving at different times, move at different speeds and by different amounts. Then there could be many ways in which our algorithm can fail. For example, suppose that a few robots of $B_j$ have seen the two terminal robots of $B_{j-1}$ with light \texttt{failed} and thereafter, have completed their moves in one go, while the rest of the robots of the batch are yet to move. Some of these robots have pending moves, while some may not move at all as they may not see the robots with light \texttt{failed} in the new configuration. Now the robots that have completed their moves earlier may now erroneously conclude that these robots belong to $B_{j+1}$, and one of them may find itself eligible to become leader. Meanwhile, one of the robots with a pending move might complete its move and finds itself eligible to become leader in the new configuration. Thus we may end up with two robots electing themselves as leader. Therefore, we have to carefully coordinate the movements of the batch so that they do not get disbanded. We have to ensure that all the robots of the batch remain vertically aligned after their moves, and also the completion of the moves of the batch as a whole must be detectable. We will need two extra colors for this. When the robots of $B_{j}$ find two terminal robots of $B_{j-1}$ with light \texttt{failed}, they will not immediately move; they will first change their lights to \texttt{ready} (See Fig. \ref{Fig: co_move}). Having all robots of $B_{j}$ with light set to \texttt{ready} will help these robots to identify their batchmates. On one hand, a robot that moves first will be able to identify the robots from its batch that are lagging behind and also detect when every one has completed their moves. On the other hand, a robot that has lagged behind will be able to remember that it has to move (from its own light \texttt{ready}) and determine how far it should move (from robots with light \texttt{ready} on its left) even if it can not see the terminal robots of batch $B_{j-1}$. Therefore, before moving, all robots of $B_{j}$ must change their lights to \texttt{ready}. But they can not verify if all their batchmates have changed their lights as they can not see all the robots of their batch. But the robots of $B_{j-1}$ are able to see all the robots of $B_{j}$, and thus can certify this. So when all the robots $B_{j}$ have changed their lights to \texttt{ready}, the terminal robots of $B_{j-1}$ will confirm this by turning their lights to \texttt{move}. Only after this, the robots of $B_{j}$ will start moving. The robots will be able to detect that the movement of the batch has completed by checking that its distance from $B_{j-1}$ is $1+\frac{1}{n}$ and there are no robots with light \texttt{ready} on its right. When it detects that the movement of the batch has completed, it will try to elect leader if it is terminal, otherwise, it will change its light to \texttt{off}.

\subsubsection{Electing Leader from Local View}

When the terminal robots of a batch will attempt to elect leader, they will require the knowledge of $\mathcal{L}$. Therefore, as different batches try to elect leader from left to right, the knowledge of $\mathcal{L}$ also needs to be propagated along the way with the help of lights. Consider the terminal robots $r^j_1$ and $r^j_2$ of a batch $B_j, j \geq 1$, that are attempting to elect a leader. In order to do so, they need two things: 1) the knowledge of $\mathcal{L}$, and 2) a full view of the next batch $B_{j+1}$. First consider the case $j = 1$. The terminal robots of the first batch $r^1_1$ and $r^1_2$ (with lights set to \texttt{terminal}) obviously have the knowledge of $\mathcal{L}$ as it is the horizontal line passing through the mid-point of the line segment joining them. Also, since $r^1_1$ and $r^1_2$ are at least $\frac{n+3}{2}$ units apart from the robots of $B_2$, they can see all the robots of $B_2$. Now suppose that a batch $B_j$, $j > 1$, is attempting to elect leader. Then as discussed in the last section, the robots of $B_j$ are horizontally exactly $1+\frac{1}{n}$ units to the right of the robots of $B_{j-1}$ and at least 2 units to the left of the robots of $B_{j+1}$. Therefore, $r^j_1$ and $r^j_2$ can see all the robots of both batches $B_{j-1}$ and $B_{j+1}$. Now since $B_{j}$ is attempting to elect leader, it implies that the first $j-1$ batches have failed to break symmetry. Hence, the first $j$ batches are symmetric with respect to $\mathcal{L}$. In particular, $\mathcal{L}$ passes through the mid-point of the line segment joining the terminal robots of $B_{j-1}$. Since $r^j_1$ and $r^j_2$ can see the terminal robots of $B_{j-1}$ (having lights set to \texttt{move}), they can determine $\mathcal{L}$.

Now, for a batch $B_j, j \geq 1$, attempting to elect leader, there are three cases to consider. If the robots of $B_{j+1}$ are asymmetric with respect to $\mathcal{L}$ (Case 1), the one of $r^j_1$ and $r^j_2$ which is in the dominant half will change its light to \texttt{switch off} and start executing \textsc{BecomeLeader()} (described in the following section). If the robots of $B_{j+1}$ are symmetric with respect to $\mathcal{L}$ and $\mathcal{L}$ passes through the center of a robot $r'$ of $B_{j+1}$ (Case 2), then $r^j_1$ and $r^j_2$ will change their lights to \texttt{symmetry}. When $r'$ finds two robots on its left batch with light \texttt{symmetry} that are equidistant from it, it will change its light to \texttt{switch off} and start executing \textsc{BecomeLeader()}. If the robots of $B_{j+1}$ are symmetric with respect to $\mathcal{L}$ and $\mathcal{L}$ does not pass through the center of any robot of $B_{j+1}$ (Case 3), then $r^j_1$ and $r^j_2$ will change their lights to \texttt{failed}. Then the robots of $B_{j+1}$ execute movements as described in the previous section.

\subsubsection{Executing \textsc{BecomeLeader()}}

 When a robot finds itself eligible to become leader, it sets its light to \texttt{switch off} and executes \textsc{BecomeLeader()} in order to fulfill all the additional conditions of a leader configuration (See Fig. \ref{Fig: become2}). A robot with light \texttt{switch off} will not do anything if it sees any robot with light other than \texttt{off} in its own batch or an adjacent batch, i.e., it will wait for those robots to turn their lights to \texttt{off} (See line \ref{code: imp} of Algorithm \ref{algo:phase2}). A robot $r$ that finds itself eligible to become leader, is either (Case 1) a terminal robot of a batch, or (Case 2) a middle robot of a batch. The first objective is to move vertically so that all robots are in $\mathcal{H} \in \{\mathcal{H}_{U}^{O}(r), \mathcal{H}_{B}^{O}(r)\}$ and at least 2 units away from $\mathcal{L}_H(r)$. In case 1, the robot has no obstruction to move vertically. But in case 2, it will have to move horizontally left first. We can show (See the proof of Theorem \ref{thm_phase2}) that it will have enough room to move and place itself at a position where there is no obstruction to move vertically. After the vertical movement, it will have to move horizontally so that all other robots are in $\mathcal{H}_{R}^{O}(r)$. But it will not try to do this in one go, as we have to also ensure that all other robots turn their lights to \texttt{off}. It will first move left to align itself with its nearest left batch, say $B_j$. From there it can see all robots of $B_{j-1}$ and $B_{j+1}$, and it will wait until all robots of $B_{j-1}$ and $B_{j+1}$ turn their lights to \texttt{off}. Then it will move to align itself with $B_{j-1}$ and so on. Eventually all the conditions of a leader configuration will be satisfied and it will change its light to \texttt{leader}.

 \clearpage


  \begin{algorithm}[H]
    \setstretch{1}
    \SetKwInOut{Input}{Input}
    \SetKwInOut{Output}{Output}
    \SetKwProg{Fn}{Function}{}{}
    \SetKwProg{Pr}{Procedure}{}{}

    \Pr{\textsc{Phase2()}}{

    $r \leftarrow$ myself
    
    \uIf{$r.light \neq$ \texttt{switch off}}{
    
    \uIf{there is a robot with light \texttt{switch off} in my batch or an adjacent batch \label{code: imp}}{$r.light \leftarrow$ \texttt{off}}
    
    \uElseIf{$r.light =$ \texttt{off} or \texttt{interior}}{

	\uIf{both terminal robots of my left batch have lights set to \texttt{failed} and the non-terminal robots (if any) have lights set to \texttt{off}}{$r.light \leftarrow$ \texttt{ready}}
	
	\ElseIf{both terminal robots of my left batch have lights set to \texttt{symmetry} and the non-terminal robots (if any) have lights set to \texttt{off}}{
	
	  \If{the two terminal robots of my left batch are equidistant from me}{
	
	    $r.light \leftarrow$ \texttt{switch off}
	
	}

	}

      }

    \uElseIf{$r.light =$ \texttt{terminal}}{
    
	\textsc{ElectLeader()}
 
      }

      \uElseIf{$r.light =$ \texttt{failed}}{
    
	\If{all robots of my right batch have their lights set to \texttt{ready}}{$r.light \leftarrow$ \texttt{move}}
 
      }

      \ElseIf{$r.light =$ \texttt{ready}}{
      
	\uIf{there is a robot $r'$ in $\mathcal{H}^{O}_{L}(r)$ with light set to \texttt{ready}}{
	
	  $d \leftarrow$ the horizontal distance of $r'$ from me
	
	  Move $d$ units towards left
	
	}
    
	\ElseIf{both terminal robots of my left batch have lights set to \texttt{move}}{
	
	  $d \leftarrow$ the horizontal distance of my left batch from me
	  
	  \uIf{$d > 1 + \frac{1}{n}$}{Move $d - 1 - \frac{1}{n}$ units towards left}
	  \ElseIf{$d = 1 + \frac{1}{n}$}{
	  
	    \If{there is no robot with light \texttt{ready} in $\mathcal{H}^{O}_{R}(r)$}{
	    
	      \uIf{I am terminal}{\textsc{ElectLeader()}}
	      \Else{$r.light \leftarrow$ \texttt{off}}
	    
	    }
	  
	  }
	
	}

      }

    }
    
    \Else{\textsc{BecomeLeader()}}
    
    }
    
     \Pr{\textsc{ElectLeader()}}{
     
	\uIf{I am in the first batch}{
	
	$\mathcal{L} \leftarrow$ the horizontal line passing through the mid-point of the line segment joining me and the other robot (with light \texttt{terminal}) on $L_V(r)$
	
	}
	\Else{
	$\mathcal{L} \leftarrow$ the horizontal line passing through the mid-point of the line segment joining the terminal robots (with lights \texttt{move}) of my left batch
	}
	\uIf{my right batch is symmetric with respect to $\mathcal{L}$}{
	
	  \uIf{$\mathcal{L}$ passes through the center of a robot of the right batch}{$r.light \leftarrow$ \texttt{symmetry}}
	  \Else{$r.light \leftarrow$ \texttt{failed}}
	
	}
	
	\ElseIf{I am in the dominant half with respect to $\mathcal{L}$ and my right batch}{$r.light \leftarrow$ \texttt{switch off}}

     }
    
\caption{Phase2}
    \label{algo:phase2} 
\end{algorithm}

\subsection{Pattern Formation from Leader Configuration}

In a leader configuration, the robots can reach an agreement on a common coordinate system. All non-leader robots in a leader configuration lie on one of the open half-planes delimited by the horizontal line passing through the leader $r_l$. This half-plane will  correspond to the positive direction of $Y$-axis or `up'. Therefore, we have an agreement on `up', `down', `left' and `right'. Now the origin will be fixed at a point such that the coordinates of $r_l$ are $(0,-2)$. Now the given pattern can be embedded on the plane with respect to the common coordinate system. Let us call these points the \emph{target points}. Order these points as $t_0, t_1, \ldots, t_{n-1}$ from top to bottom, and from right to left in case multiple robots on the same horizontal line (See Fig. \ref{fig: stage2_robot_target}). Order the robots as $r_l = r_0, r_1, \ldots, r_{n-1}$ from bottom to up, and from left to right in case multiple robots on the same horizontal line (See Fig. \ref{fig: stage2_robot_order}). The non-leader robots will move sequentially according to this order and place themselves on $\mathcal{L}_H(r_l)$. Then sequentially $r_1, \ldots, r_{n-1}$ will move to the target points $t_0, \ldots, t_{n-2}$, and finally $r_0$ will move to $t_{n-1}$. Pseudocode of the algorithm is given in Algorithm \ref{main_algorithm: stage 2} and further details are in Appendix \ref{appendix_s2}.

\begin{minipage}{1\textwidth}
  \begin{algorithm}[H]
    \setstretch{.1}
    \SetKwInOut{Input}{Input}
    \SetKwInOut{Output}{Output}
    \SetKwProg{Fn}{Function}{}{}
    \SetKwProg{Pr}{Procedure}{}{}

  \Input{The configuration of robots visible to me.}
    
    \Pr{\textsc{PatternFormationFromLeaderConfiguration()}}{

    $r \leftarrow$ myself
    
    $r_l \leftarrow$ the robot with light \texttt{leader}
    
    \uIf{$r.light =$ \texttt{off}}{
    
      \uIf{($r_l \in \mathcal{H}_B^O(r)$)  and (there is no robot in $\mathcal{H}_{B}^{O}(r) \cap \mathcal{H}_{U}^{O}(r_l)$) and ($r$ is leftmost on $\mathcal{L}_H(r)$) \label{code s2: 0}}{
    
	  \uIf{there are no robots on $\mathcal{L}_H(r_l)$ other than $r_l$ \label{code s2: 1}}{
	  
	    \uIf{there is a robot with light \texttt{done}\label{code s2: 2}}{
	    
	    \uIf{I am at $t_{n-2}$}{$r.light \leftarrow$ \texttt{done}}
	    \Else{Move to $t_{n-2}$}
	    
	    }
	    \Else{Move to $(1,-2)$}
	  
	  }
	  \uElseIf{there are $i$ robots on $\mathcal{L}_H(r_l)$ other than $r_l$ at $(1,-2), \ldots, (i,-2)$\label{code s2: 3}}{Move to $(i+1,-2)$}
	  \ElseIf{there are $i$ robots on $\mathcal{L}_H(r_l)$ other than $r_l$ at $(n-i,-2), \ldots, (n-1,-2)$\label{code s2: 6}}{
	  
	    \uIf{I am at $t_{n-i-2}$}{$r.light \leftarrow$ \texttt{done}}
	    \Else{Move to $t_{n-i-2}$}
	  
	  }
    
      }
      
      \ElseIf{$r_l \in \mathcal{L}_H(r)$ and $\mathcal{H}_U^O(r)$ has no robots with light \texttt{off}\label{code s2: 4}}{
      
	\If{I am at $(i,-2)$}{Move to $t_{i-1}$}
      
      }

    }

    \ElseIf{$r.light =$ \texttt{leader}}{
    
      \If{there are no robots with light \texttt{off} \label{code s2: 5}}{
    
    \uIf{I am at $t_{n-1}$}{$r.light \leftarrow$ \texttt{done}}
	    \Else{Move to $t_{n-1}$}}
    
      }
    }

\caption{Pattern Formation from Leader Configuration}
    \label{main_algorithm: stage 2} 
\end{algorithm}

\end{minipage}

\section{Conclusion}

Using 4 extra colors, we have extended the results of \cite{BoseKAS19} to the more realistic setting of fat robots with non-rigid movements and also improved the move complexity to $\Theta(n)$, which is asymptotically optimal. Techniques used in Phase 2 of our algorithm can be used to solve \textsc{Leader Election} without movement for luminous opaque point robots for any initial configuration where \textsc{Leader Election} is solvable in full visibility model, except for the configuration where all robots are collinear. An interesting question is whether there is a no movement \textsc{Leader Election} algorithm for (luminous and opaque) fat robots. Another open question is whether it is possible to solve $\mathcal{APF}$ for opaque (point or fat) robots with only agreement in chirality.

\bibliographystyle{plain}
\bibliography{fat_pattern}

\appendix
\newpage

 \section{Correctness of Stage 1}\label{appendix_s1}
 
 For the correctness proofs, we shall use the notion of a \emph{stable configuration} from \cite{BoseKAS19}. Denote the position of a robot $r$ at time $t$ to be $r(t)$. Suppose that a robot $r$ at  $p$ takes a snapshot at time $t_1$. Based on this snapshot, suppose that it decides to change its light (Case 1) or move to a different point (Case 2) or both (Case 3). In case 1, assume that it changes its light at time $t_2 > t_1$. In case 2, assume that it starts moving at time $t_3 > t_1$. When we say that it starts moving at $t_3$, we shall mean that $r(t_3) = p$, but $r(t_3 + \epsilon) \neq p$ for sufficiently small $\epsilon > 0$. For case 3, assume that $r$ changes its light at  $t_2 > t_1$ and starts moving at  $t_3 > t_2$. Then we say that $r$ has a \emph{pending move} at $t$ if $t \in (t_1, t_2)$ in case 1 or $t \in (t_1, t_3]$ in case 2 and 3. A robot $r$ is said to be \emph{stable} at time $t$, if $r$ is stationary and has no pending move at $t$. A configuration at time $t$ is said to be a \emph{stable configuration} if every robot is stable at $t$. Also, a configuration at time $t$ is said to be a \emph{final configuration} if 
\begin{enumerate}                                                                                                                                                                                                                                                                                                \item every robot at $t$ is stable,                                                                                                                                                                                                                                                                                                                                                                                                                                                                                                                                                                                   \item any robot taking a snapshot at $t$ will not decide to move or change its color.                                                                                                                                                                                                                                                                                                 \end{enumerate}

The aim of the leader election stage is to form a \emph{leader configuration}. Formally, we define a leader configuration to be a stable configuration in which there is a unique robot $r_l$ such that
\begin{enumerate}
   \item $r_l.light =$ \texttt{leader},
   
   \item $r.light =$ \texttt{off} for all $r \in$ $\mathcal{R} \setminus \{r_l\}$,
   
   \item there is $\mathcal{H} \in \{\mathcal{H}_{U}^{O}(r_l), \mathcal{H}_{B}^{O}(r_l)\}$ such that $r \in \mathcal{H}_{R}^{O}(r_l) \cap \mathcal{H}$ for all $r \in$ $\mathcal{R} \setminus \{r_l\}$,
   
   \item distance of any robot of $\mathcal{R} \setminus \{r_l\}$ from $\mathcal{L}_H(r_l)$ is at least 2 units.
\end{enumerate}
 
 In this section, we shall prove Theorem \ref{thm_s1}. It will follow from Theorem \ref{thm_phase11}, Theorem \ref{thm_phase12}, Theorem \ref{thm_phase13} and Theorem \ref{thm_phase2}, whose proofs are given in the following subsections.

 \begin{theorem}\label{thm_s1}
 For any initial configuration $\mathbb{C}(0)$ which is not symmetric with respect to a line $\mathcal{K}$ such that 1) $\mathcal{K}$ is parallel to the $X$-axis and 2) $\mathcal{K}$ is not passing through the center of any robot,  $\exists$ $T_1 > 0$ such that $\mathbb{C}(T_1)$ is a leader configuration.
 \end{theorem}

 \subsection{Correctness of Phase 1}\label{appendix_p1}

 \begin{theorem}\label{thm_phase11}
 If the first batch of $\mathbb{C}(0)$ has exactly one robot, then $\exists$ $T_1 > 0$ such that $\mathbb{C}(T_1)$ is a leader configuration. 
 \end{theorem}
 
 \begin{proof}
  Let $r$ be the robot in the first batch. From its local view, $r$ will find that it is the unique leftmost in the configuration and start executing \textsc{BecomeLeader()}, i.e., it will move vertically until all other robots are in $\mathcal{H}_{R}^{O}(r) \cap \mathcal{H}$, where $\mathcal{H} \in \{\mathcal{H}_{U}^{O}(r), \mathcal{H}_{B}^{O}(r)\}$ and the distance of any robot of $\mathcal{R} \setminus \{r\}$ from $\mathcal{L}_H(r)$ is at least 2 units. 
  
  \textbf{Case 1} If there is no obstruction to move vertically in both directions (`up' and `down'), it will move `down' according to its local coordinate system. 
  
  \textbf{Case 2} If there is exactly one direction where there is no obstruction to move vertically, then it will move in that direction accordingly. 
  
  \textbf{Case 3} If there is obstruction in both directions, then it will have to first move left (See Fig. \ref{Fig: proof0}). 
  
  When the required condition is achieved, $r$ will change its light to \texttt{leader} at some time $T_1 > 0$. It is easy to see that all other robots will remain stable in $[0, T_1]$. So, $\mathbb{C}(T_1)$ is the required leader configuration.  \qed
 \end{proof}

 \begin{figure}[h!]
\centering
\subcaptionbox[Short Subcaption]{
       \label{}
}
[
    0.48\textwidth 
]
{
    \fontsize{8pt}{8pt}\selectfont
    \def\svgwidth{0.48\textwidth}
    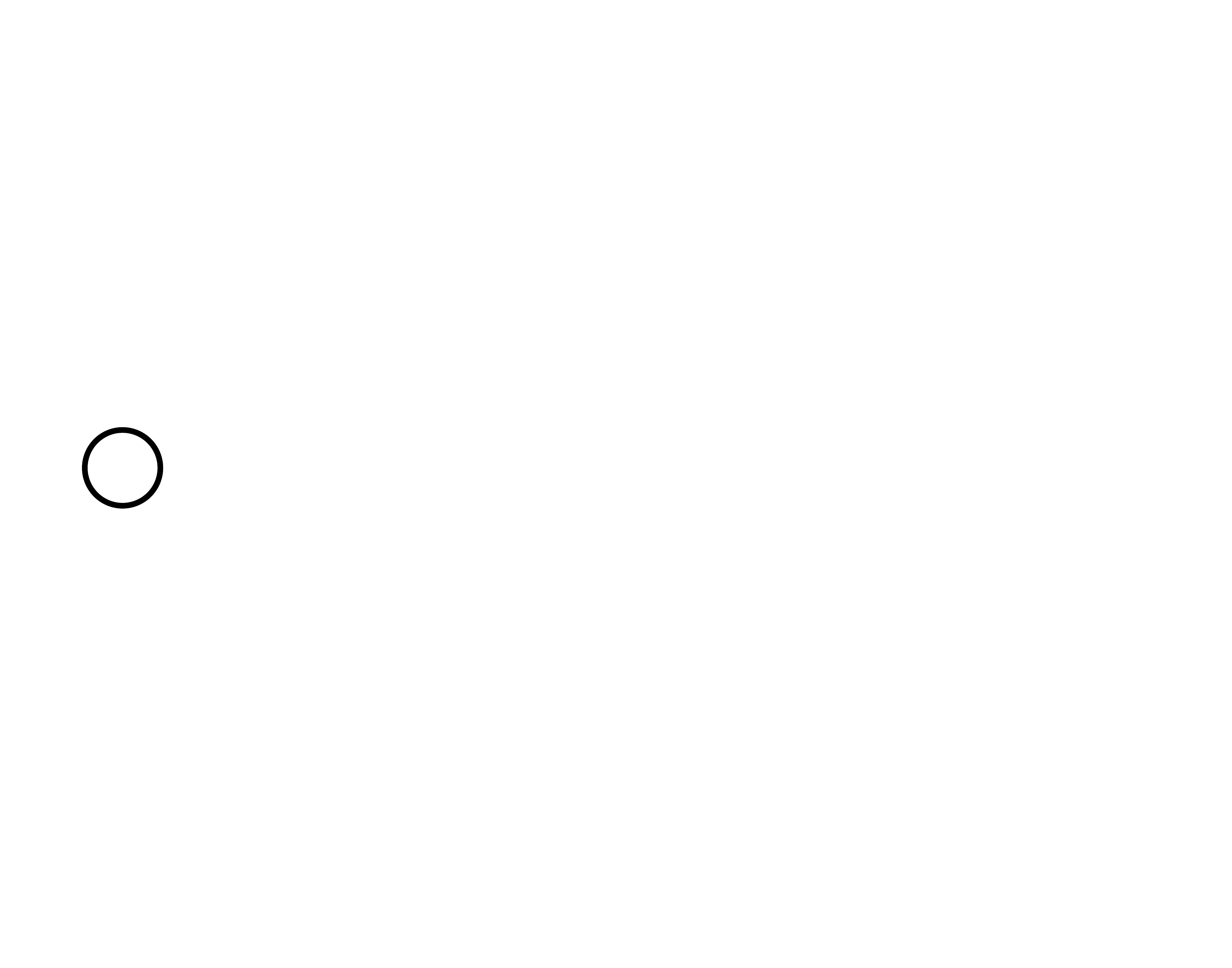
}
\hfill
\subcaptionbox[Short Subcaption]{
     \label{}
}
[
    0.48\textwidth 
]
{
    \fontsize{8pt}{8pt}\selectfont
    \def\svgwidth{0.48\textwidth}
    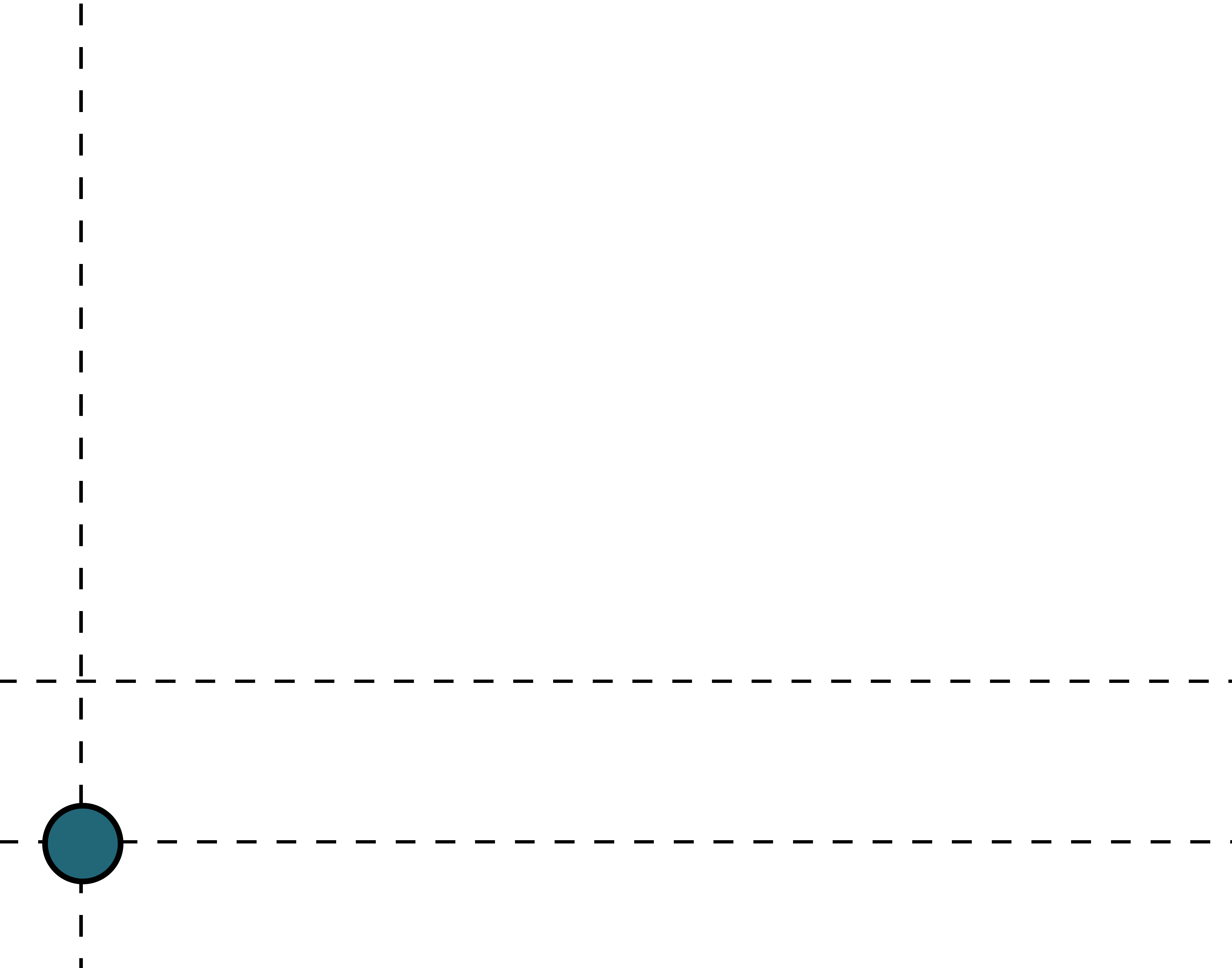
}
\caption[Short Caption]{a) Execution of \textsc{BecomeLeader()} by the unique leftmost robot $r$. b) The leader configuration after $r$ changes its light to \texttt{leader}.}
\label{Fig: proof0}
\end{figure}

 \begin{theorem}\label{thm_phase12}
 If the first batch of $\mathbb{C}(0)$ has exactly two robots, then $\exists$ $T' > 0$ such that $\mathbb{C}(T')$ is a stable configuration where 
 \begin{enumerate}
    
   \item the first batch has exactly two robots with lights set to \texttt{terminal},
   
   \item the robots of all other batches have their lights set to \texttt{off},
   
   \item the distance between the first and the second batch is at least $\frac{n+3}{2}$ units.
\end{enumerate}
 
\end{theorem}
  
  \begin{proof}
   
   Let $r_1$ and $r_2$ be the two robots of the first batch of $\mathbb{C}(0)$. After some time, both of them will change their lights to \texttt{terminal}. Notice that the robots do not move until both of them change their lights. If the distance $d$ between the first and the second batch is already at least $\frac{n+3}{2}$ in $\mathbb{C}(0)$, then the robots do not need to move and we are done. So let $d < \frac{n+3}{2}$. Then the robots $r_1$ and $r_2$ will have to move $\frac{n+3}{2} - d$ units to the left, to the points $P_1$ and $P_2$ respectively, so that their distance from the second batch becomes exactly $\frac{n+3}{2}$. We will show that the algorithm will successfully bring the robots stationary at the points $P_1$ and $P_2$.

   Suppose that the robots take snapshots at times $t_1 \leq t_2 \leq \ldots$. We shall show that if $r_i  \in \{r_1,r_2\}$ takes a snapshot at $t_k$, then it will not decide (instructed by Algorithm \ref{algo:phase1}) to move beyond $P_i$. Without loss of generality, assume that $r_1$ takes snapshot at $t_1$. Notice that if the second batch is too  close to the first batch, $r_1$ and $r_2$ may not be able to correctly identify the second batch (See Fig. \ref{view}). So $r_1$ will decide to move $\frac{n+3}{2} - d_1$ units left, where $d_1$ is its horizontal distance from the leftmost robot (that it can see) on its right. Notice that $d_1 \geq d$ (where $d$ is its distance from the actual second batch), and hence, $\frac{n+3}{2} - d_1 \leq \frac{n+3}{2} - d$. So $r_1$ does not decide to move beyond $P_1$. Suppose that the same is true up to the $(k-1)$th snapshot. Suppose $r_i  \in \{r_1,r_2\}$ takes the snapshot at $t_k$ and at that time, its distance from the (actual) second batch of $\mathbb{C}(0)$ is $d'$.

\begin{figure}
   \centering
   \fontsize{8pt}{8pt}\selectfont
   \def\svgwidth{0.25\textwidth}
   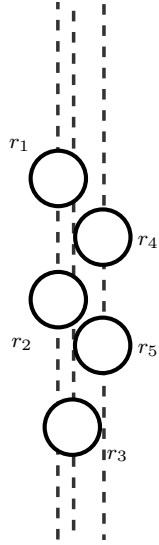
   \caption{The first three batches of a configuration, where $B_1 = \{r_1, r_2\}$, $B_2 = \{r_3\}$ and $B_3 = \{r_4, r_5\}$. Here, $r_1$ can not see $r_3$. Therefore, $r_1$ thinks that $\{r_4, r_5\}$ is the second batch.}
   \label{view}
\end{figure}

    \textbf{Case A1} The other robot $r_j$ (with light set to \texttt{terminal}) is on $\mathcal{L}_V(r_i)$. If $d' = \frac{n+3}{2}$, it will not move. Otherwise, by the same argument as earlier, $r_i$ will not decide to move beyond $P_i$.

%
    
    \textbf{Case A2} The other robot $r_j$ (with light set to \texttt{terminal}) is on $\mathcal{H}_{R}^{O}(r_i)$. In this case, $r_i$ will decide not to move.
    
    \textbf{Case A3} The other robot $r_j$ (with light set to \texttt{terminal}) is on $\mathcal{H}_{L}^{O}(r_i)$. In this case, $r_i$ will decide to move left to vertically align itself with $r_j$.  By our assumption, $r_j$ is not beyond $P_j$. So $r_i$ has not decided to move beyond $P_i$.

   Now we show that if $r_i$ at some time $t$ is stationary and its distance $d'$ from the second batch (of $\mathbb{C}(0)$) is less than $\frac{n+3}{2}$, then it will decide to move towards $P_i$ at some time after $t$. Suppose that it takes the first snapshot after $t$ at $t_k$.
   
    \textbf{Case B1} The other robot $r_j$ (with light set to \texttt{terminal}) is on $\mathcal{L}_V(r_i)$. Of course, it will decide to move left in this case.

%
    
    \textbf{Case B2} The other robot $r_j$ (with light set to \texttt{terminal}) is on $\mathcal{H}_{R}^{O}(r_i)$. In this case, $r_i$ will decide not to move. Then after finite time, $r_j$ will overtake $r_i$ or stop on $\mathcal{L}_V(r_i)$. In either case, $r_i$ will decide to move left.

    \textbf{Case B3} The other robot $r_j$ (with light set to \texttt{terminal}) is on $\mathcal{H}_{L}^{O}(r_i)$. Clearly, $r_i$ will decide to move left to vertically align itself with $r_j$.

   So, we have shown that each $r_i, i = 1, 2$ gradually moves towards $P_i$ and never decides to move beyond $P_i$. Therefore both will eventually reach $P_i$ and remain stationary. Also observe that during this process, all the other robots remain stable. Therefore, we shall obtain the required configuration at some time $T' > 0$. \qed
   
  \end{proof}

 \begin{theorem}\label{thm_phase13}
 If the first batch of $\mathbb{C}(0)$ has more than two robots, then $\exists$ $T' > 0$ such that $\mathbb{C}(T')$ is a stable configuration where 
 \begin{enumerate}
    
   \item the first batch has exactly two robots with lights set to \texttt{terminal},
   
   \item the robots of the second batch have their lights set to \texttt{interior},
   
   \item the robots of all other batches have their lights set to \texttt{off},
   
   \item the distance between the first and the second batch is exactly $\frac{n+3}{2}$ units.
\end{enumerate}
 
\end{theorem}

  \begin{proof}
   Let $r_1$ and $r_2$ be the two terminal robots of the first batch of $\mathbb{C}(0)$. We will first show that both robots will start moving after some finite time. Let $r_1'$ and $r_2'$ be the two non-terminal robots of the first batch that are adjacent to $r_1$ and $r_2$ respectively ($r_1' = r_2'$, if there are exactly three robots in the batch). It is easy to see that one of the robots, say $r_i$ will start moving. It suffices to argue for the situation where $r_i$ starts moving before the other terminal robot $r_j$ wakes up. Notice that since $r_i$ starts moving, $r_i'$ must have set its light to \texttt{interior}. Then the algorithm ensures that (See line \ref{code: p1_0} of Algorithm \ref{algo:phase1}) eventually $r_j'$ must also change its light to \texttt{interior}. So when $r_j$ takes snapshot, it finds that the conditions in line \ref{code: p1_1} of Algorithm \ref{algo:phase1} are satisfied, and hence will decide to move.
   
   Notice that in this case, our algorithm asks the robots to move $\frac{n+3}{2}$ units. This is because, in this case, the non-terminal robots of the first batch would be the second batch of the desired configuration. Arguing similarly as in the proof of Theorem \ref{thm_phase11}, we can show that they will be able to do so. Also, it is easy to see that all non-terminal robots of the first batch of $\mathbb{C}(0)$ will eventually turn their lights to \texttt{interior}, but will stay stationary throughout the process. All the other robots will also remain stable during the process. Hence, we shall obtain the required configuration at some time $T' > 0$. \qed

  \end{proof}

\subsection{Correctness of Phase 2}\label{appendix_p2}

  \begin{theorem}\label{thm_phase2}
 If $\mathbb{C}(T')$ is the configuration from Theorem \ref{thm_phase11} or Theorem \ref{thm_phase12}, then $\exists$ $T_1 > T'$ such that $\mathbb{C}(T_1)$ is a leader configuration. 
 
\end{theorem}

\begin{proof}
 We assumed that $\mathbb{C}(0)$ is not symmetric with respect to any line that is parallel to the $X$-axis and does not pass through the center of any robot. The same should be true for $\mathbb{C}(T')$ because it is obtained from $\mathbb{C}(0)$ by moving the two terminal robots of the first batch horizontally by equal amounts. So, in particular, if $\mathcal{L}$ is the horizontal line passing through the mid-point of the line segment joining the terminal robots of the first batch of $\mathbb{C}(T')$, then either  $\mathbb{C}(T')$ is asymmetric with respect to $\mathcal{L}$ or $\mathbb{C}(T')$ is symmetric with respect to $\mathcal{L}$, but there is a robot whose center lies on $\mathcal{L}$. So there is at least one batch that is either asymmetric with respect to $\mathcal{L}$ or symmetric with respect to $\mathcal{L}$, but it has a robot whose center lies on $\mathcal{L}$. Let $B_i$ be the first such batch. Obviously, $i > 1$.

  We have discussed in Section \ref{sec: p2} how different batches from left to right will sequentially try to elect a leader. Obviously, $B_{i-1}$ will be successful. If (Case 1) $B_{i}$ is asymmetric with respect to $\mathcal{L}$, then one of the terminal robots of $B_{i-1}$ will start executing \textsc{BecomeLeader()}, while if (Case 2) $B_{i}$ is symmetric with respect to $\mathcal{L}$, then the terminal robots of $B_{i-1}$ will change their lights to \texttt{symmetry} and the robot $r \in B_i$ lying on $\mathcal{L}$ will start executing \textsc{BecomeLeader()}.

  We have to show that each of $B_{1}, \ldots, B_{i-1}$ have enough space to successfully execute the instructions of the algorithm.
  We shall show that when the terminal robots of $B_{j}, j \geq 1$ call \textsc{ElectLeader()}, the distance between $B_{j}$ and $B_{j+1}$ will be at least 2 units. This will ensure the following:
  
  \begin{enumerate}
   \item the robots of $B_{j}$ can fully see the batch $B_{j+1}$,
   
   \item if the terminal robots of $B_{j}$ fail to elect leader and turn their lights to \texttt{failed}, the robots of $B_{j+1}$ will be able to see this and also will be able to move left and place themselves $1 + \frac{1}{n}$ units apart from $B_{j}$,
   
   \item if the middle robot of $B_{j+1}$ is to become leader, it can move 1 unit to the left so that there is no obstruction to move vertically. 
  \end{enumerate}

  This is obvious for $j = 1$, because $B_1$ is $\frac{n+3}{2} > 2$ units away from $B_2$. Now consider a batch $B_j, 1 < j \leq i-1$ calling \textsc{ElectLeader()}. Initially the distance between $B_1$ and $B_2$ was at least $\frac{n+3}{2}$ and then each batch $B_l, 1 < l \leq j$, have moved left and placed themselves exactly $1 + \frac{1}{n}$ units apart from $B_{l-1}$. This implies that the batch $B_{j}$ has moved at least $\frac{n+3}{2} - (j-1)(1+\frac{1}{n})$. So, after the movement, the distance between $B_{j}$ and $B_{j+1}$ is at least $\frac{n+3}{2} - (j-1)(1+\frac{1}{n})$. Now note that each $B_1, \ldots, B_{j}$ is symmetric with respect $\mathcal{L}$ and center of no robot of these batches lie on $\mathcal{L}$. Therefore, each $B_1, \ldots, B_{j}$ has at least 2 robots. So, $2j < n \Rightarrow j < \frac{n}{2}$. Therefore, the required distance between $B_{j}$ and $B_{j+1}$ is 
  \begin{alignat*}{2}
   \geq & \frac{n+3}{2} - (j-1)(1+\frac{1}{n}) \\
   > & \frac{n+3}{2} - (\frac{n}{2} - 1)(1+\frac{1}{n}) \\
   = & 2 + \frac{1}{n}
  \end{alignat*}

  When a robot finds itself eligible to become leader, it sets its light to \texttt{switch off} and starts executing \textsc{BecomeLeader()}. We have not given any pseudocode for \textsc{BecomeLeader()}. The execution of \textsc{BecomeLeader()} for a robot with light \texttt{off} was described in the proof of Theorem \ref{thm_phase11}.  An informal description of the process for a robot with light \texttt{switch off} was given in Section \ref{sec: p2}. 
  
  In Case 1, a terminal robot $r$ of $B_{i-1}$ will start executing \textsc{BecomeLeader()}. Obviously there is one vertical direction for $r$ to move without any obstruction. However, in Case 2, unless $B_{i}$ has exactly one robot, the middle robot $r$ of $B_{i}$ has both vertical directions blocked. So, $r$ will move 1 unit to the left and there will be no obstruction to move vertically. So in either case, $r$ will move vertically and eventually all other robots will be in $\mathcal{H} \in \{\mathcal{H}_{U}^{O}(r), \mathcal{H}_{B}^{O}(r)\}$ and the distance of any robot of $\mathcal{R} \setminus \{r\}$ from $\mathcal{L}_H(r)$ will be at least 2 units.  Then $r$ will sequentially move to the central axes of the batches on its left and all robots that having light set to any color other than \texttt{off}, will turn them back to \texttt{off}.

  However, there can be a complication regarding the first batch $B_1$ because, as we shall see, there is a possibility that they might start executing the algorithm for Phase 1. Notice that when $r$ starts executing \textsc{BecomeLeader()}, the two robots $r^1_1$ and $r^1_2$ of $B_1$ have their lights set to \texttt{move}. Notice that they will change their light to \texttt{off}, only when $r$ aligns itself with $B_2$ at some time $t_1$ (See line \ref{code: imp} of Algorithm \ref{algo:phase2}). So, upto time $t_1$, $r^1_1$ and $r^1_2$ with lights \texttt{move} will remain stable. After $t_1$, they will change their lights to \texttt{off}. After this, $r$ will start moving to the left. At some time $t_2$, $r$ will become aligned with $r^1_1$ and $r^1_2$. Clearly, $r^1_1$ and $r^1_2$ will remain stable in $[t_1, t_2]$. In $\mathbb{C}(t_2)$, one of $r^1_1$ and $r^1_2$, say $r^1_1$, can not see $r$. So, if $r^1_1$ takes a snapshot at $t_2$, it will decide to turn its light to \texttt{terminal} according to Algorithm \ref{algo:phase1}. However, $r^1_2$ will stay stable as it does not loose sight of $r$ with light \texttt{switch off}. Hence, it does not execute Algorithm \ref{algo:phase1}.  Recall that a robot with light \texttt{terminal} decides to move only if it finds another robot with light \texttt{terminal} or \texttt{interior}. So, $r^1_1$ will not move. When $r$ moves further left, $r^1_1$ will again see it and will change its light to \texttt{off}. Therefore, we shall obtain a leader configuration at some time $T_1$. \qed
 
\end{proof}

 \clearpage

 \begin{figure}[p]
\centering
\subcaptionbox[Short Subcaption]{ The batch $B_3$ is $1 + \frac{1}{n}$ units to the right of $B_2$. The terminal robots $r^3_1$, $r^3_2$ of $B_3$ have change their lights to \texttt{failed}. However, the robot $r$ of $B_3$ still has its light set to \texttt{ready}. 
     \label{}
}
[
    0.48\textwidth 
]
{
    \fontsize{8pt}{8pt}\selectfont
    \def\svgwidth{0.48\textwidth}
    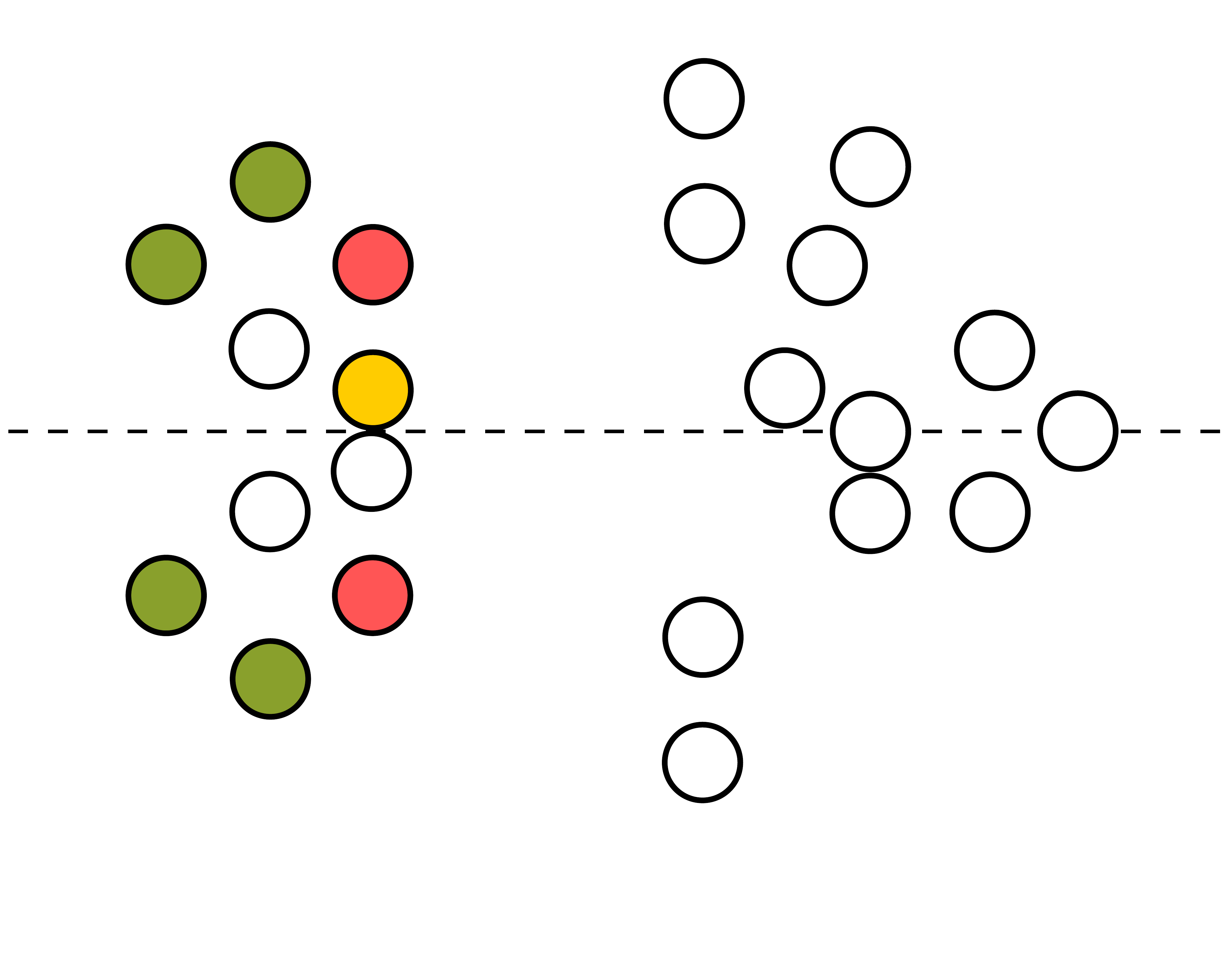
}
\hfill
\subcaptionbox[Short Subcaption]{The robots of $B_4$ waits until $r$ changes its light to \texttt{off}.
     \label{}
}
[
    0.48\textwidth 
]
{
    \fontsize{8pt}{8pt}\selectfont
    \def\svgwidth{0.48\textwidth}
    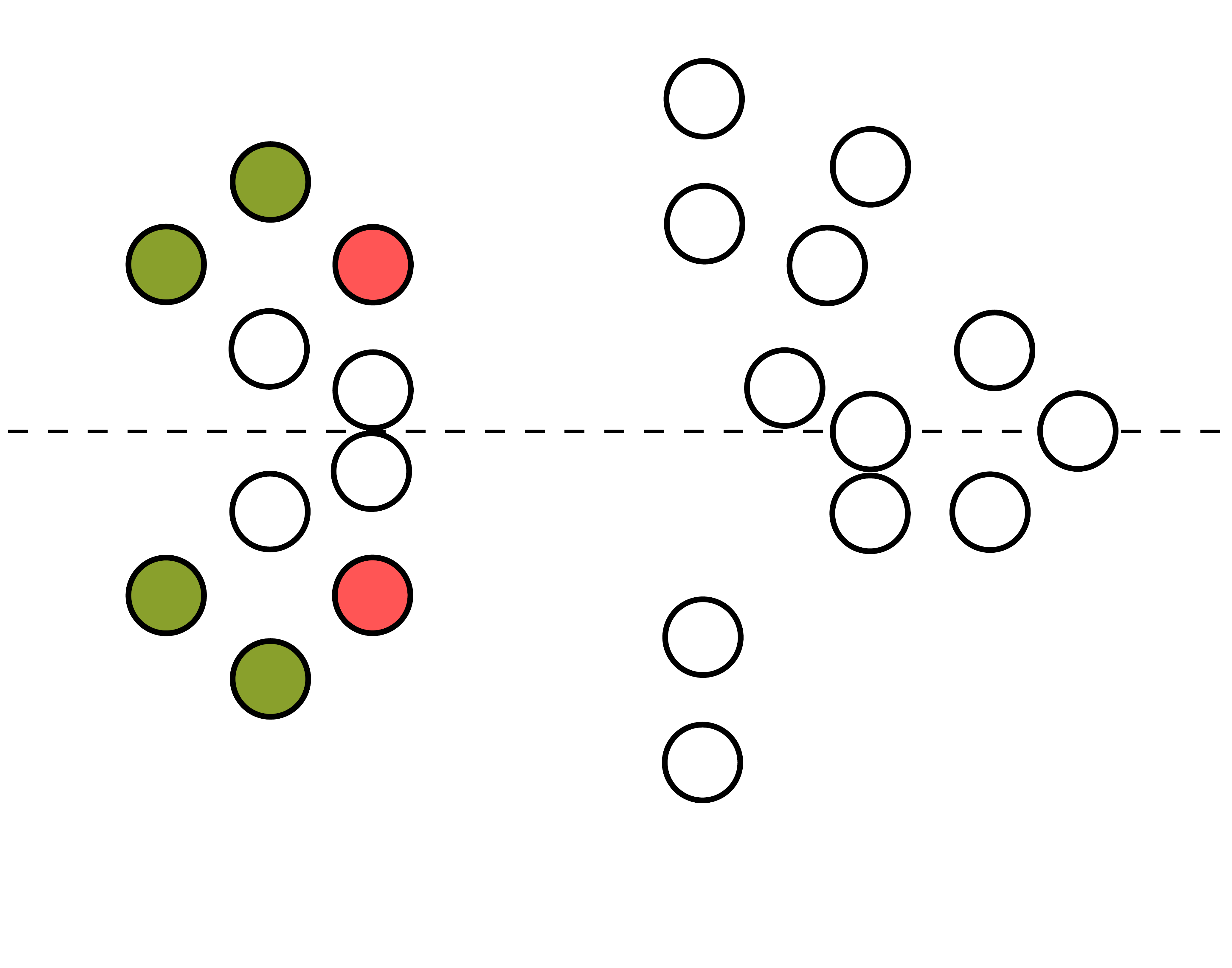
}
\hfill
\subcaptionbox[Short Subcaption]{The robots of $B_4$ are changing their lights to \texttt{ready}. The terminal  robots $r^3_1$, $r^3_2$ of $B_3$ will wait for all the robots of $B_4$ to change their lights.
       \label{}
}
[
    0.48\textwidth 
]
{
    \fontsize{8pt}{8pt}\selectfont
    \def\svgwidth{0.48\textwidth}
    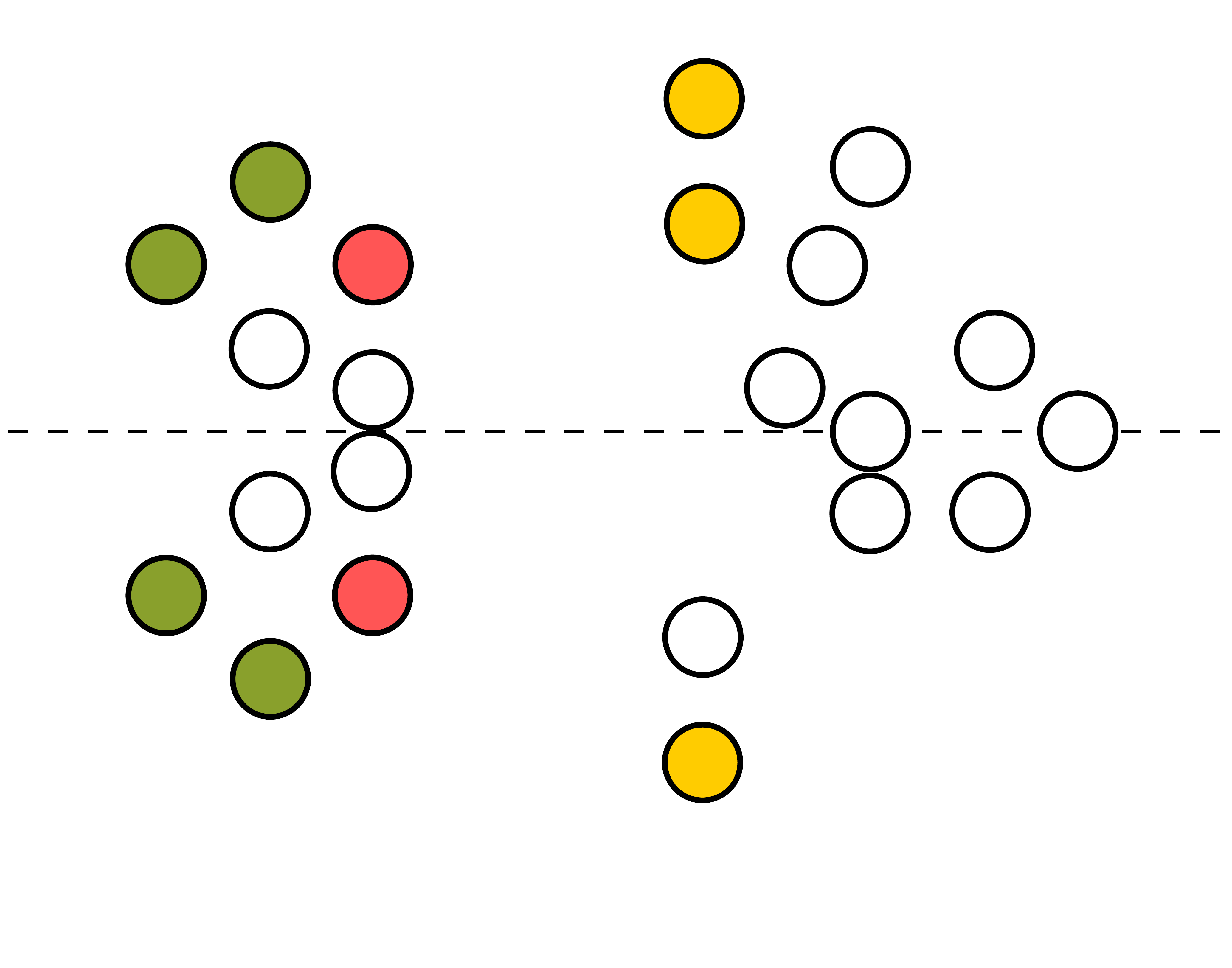
}
\hfill
\subcaptionbox[Short Subcaption]{
     After all the robots of $B_4$ change their lights to \texttt{ready},  $r^3_1$ and $r^3_2$ change their lights to \texttt{move}.
}
[
    0.48\textwidth 
]
{
    \fontsize{8pt}{8pt}\selectfont
    \def\svgwidth{0.48\textwidth}
    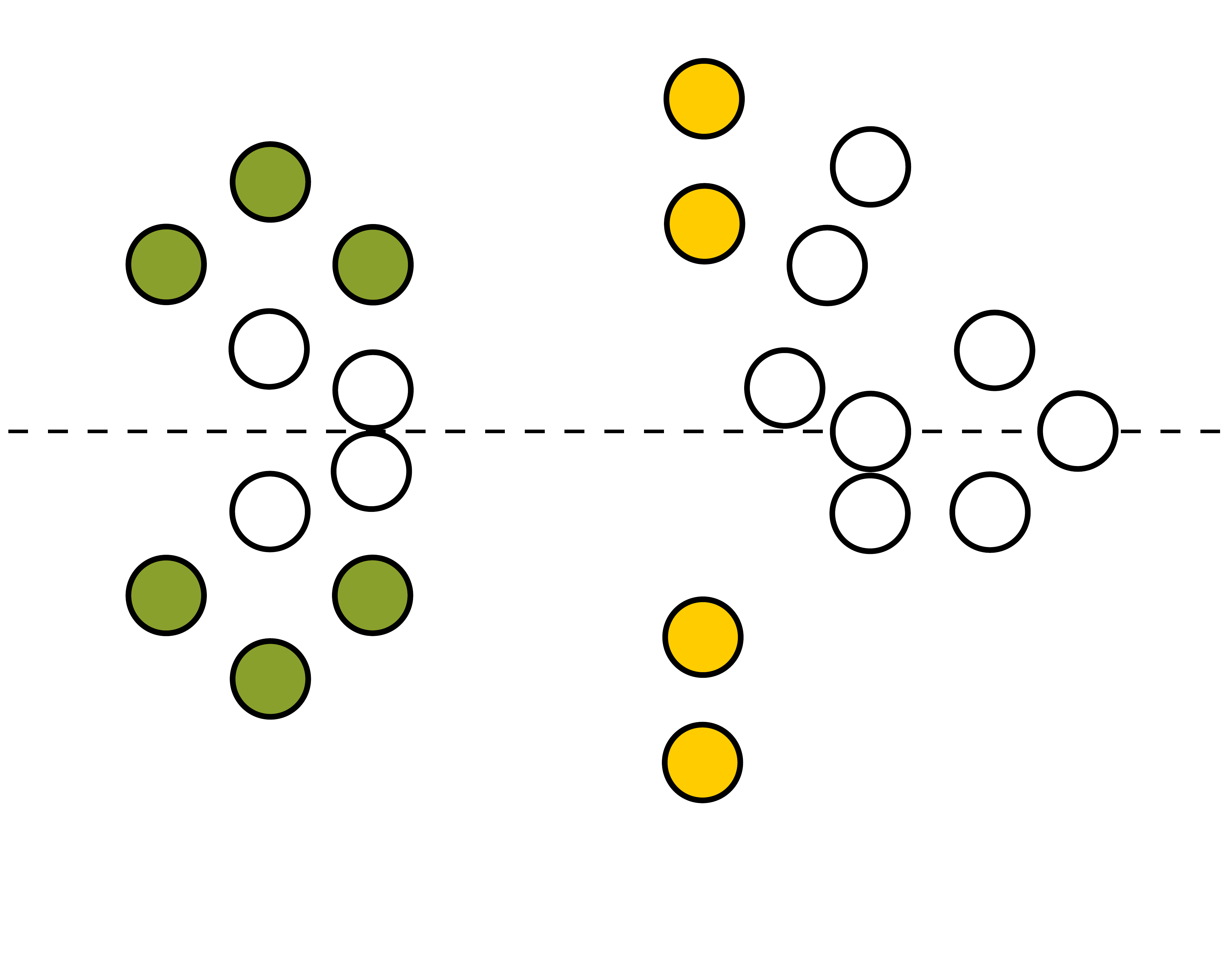
}
\hfill

\subcaptionbox[Short Subcaption]{After $r^3_1$ and $r^3_2$ change their lights to \texttt{move}, the robots of $B_4$ start moving. Here $r''$ has not yet started. If it takes a snapshot of this configuration, it will decide to align itself with $r'$. 
     \label{}
}
[
    0.48\textwidth 
]
{
    \fontsize{8pt}{8pt}\selectfont
    \def\svgwidth{0.48\textwidth}
    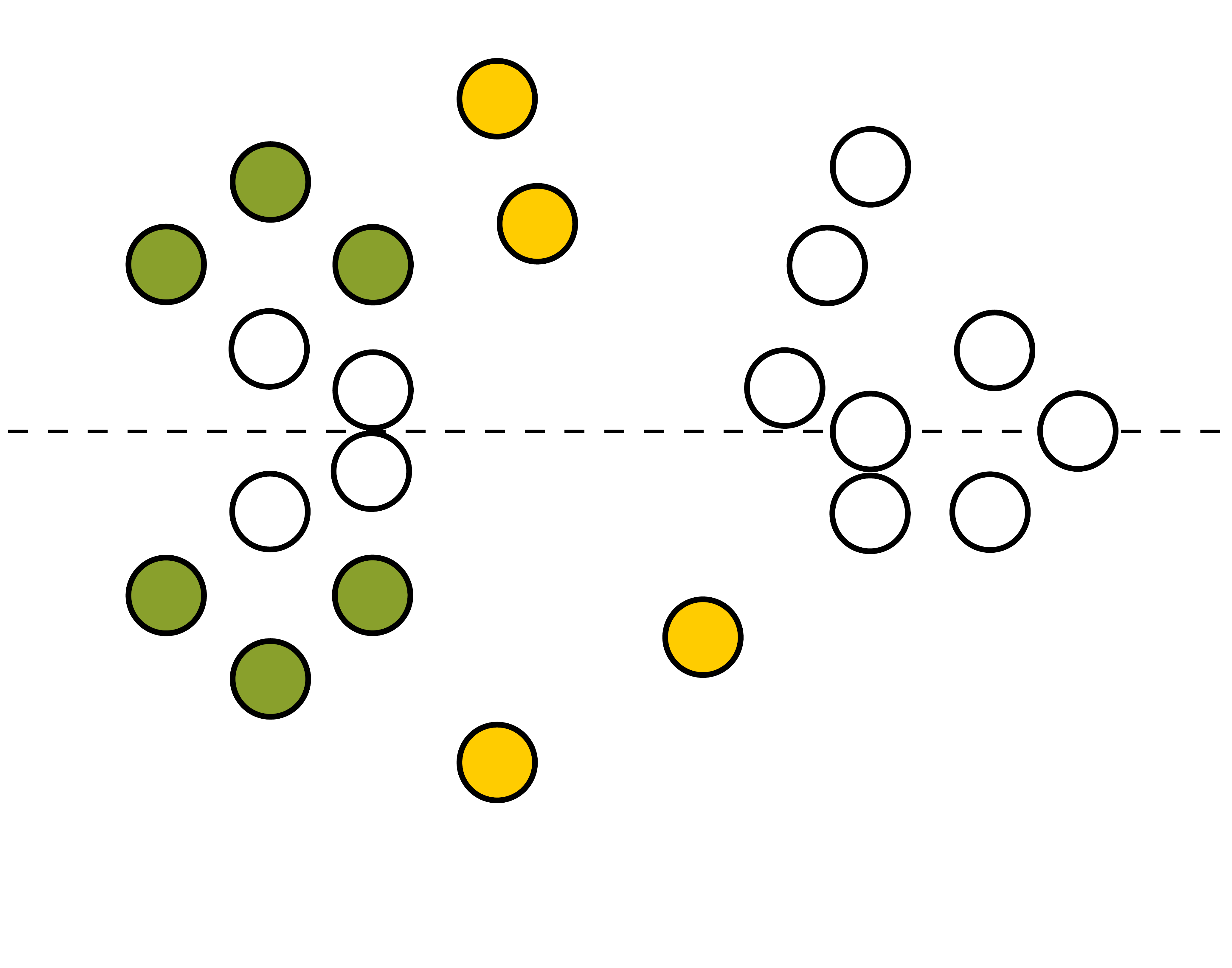
}
\hfill
\subcaptionbox[Short Subcaption]{ Eventually the robots of $B_4$ will stop $1 + \frac{1}{n}$ units to the right of $B_3$.
     \label{}
}
[
    0.48\textwidth 
]
{
    \fontsize{8pt}{8pt}\selectfont
    \def\svgwidth{0.48\textwidth}
    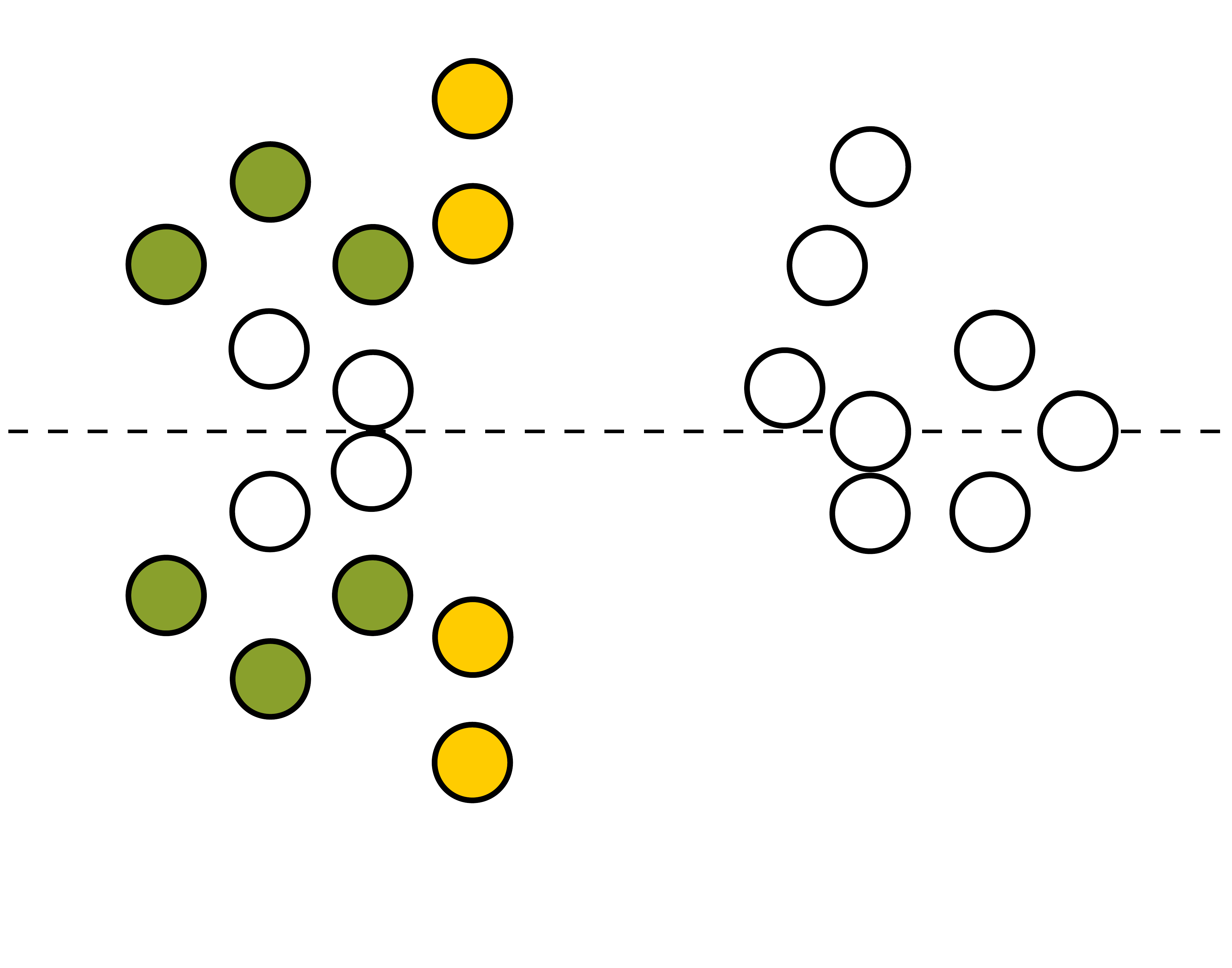
}
\caption[Short Caption]{Coordinated movement of a batch in Phase 2.}
\label{Fig: co_move}
\end{figure}

 \begin{figure}[p]
\centering
\subcaptionbox[Short Subcaption]{ (Case 1) $r$ finds itself in the dominant half and changes its light to \texttt{switch off}. Two robots in its batch, $r'$ and $r''$, still have their lights set to \texttt{ready}. 
     \label{}
}
[
    0.48\textwidth 
]
{
    \fontsize{8pt}{8pt}\selectfont
    \def\svgwidth{0.48\textwidth}
    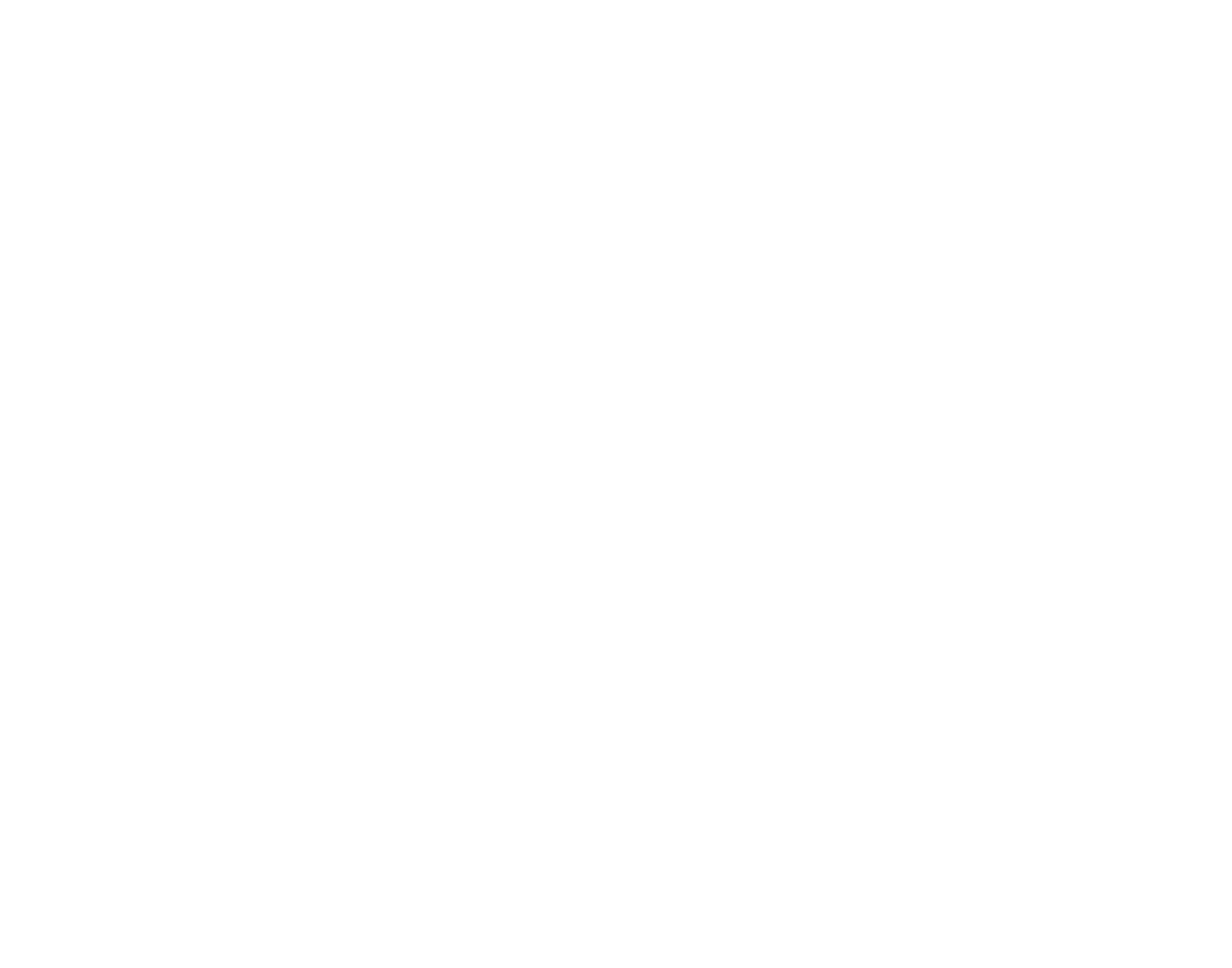
}
\hfill
\subcaptionbox[Short Subcaption]{ (Case 1) $r$ starts moving after $r^3_1$, $r^3_2$ and $r'$ have changed their lights to \texttt{off}. But $r$ can not see $r''$ which is yet to change its light.
     \label{}
}
[
    0.48\textwidth 
]
{
    \fontsize{8pt}{8pt}\selectfont
    \def\svgwidth{0.48\textwidth}
    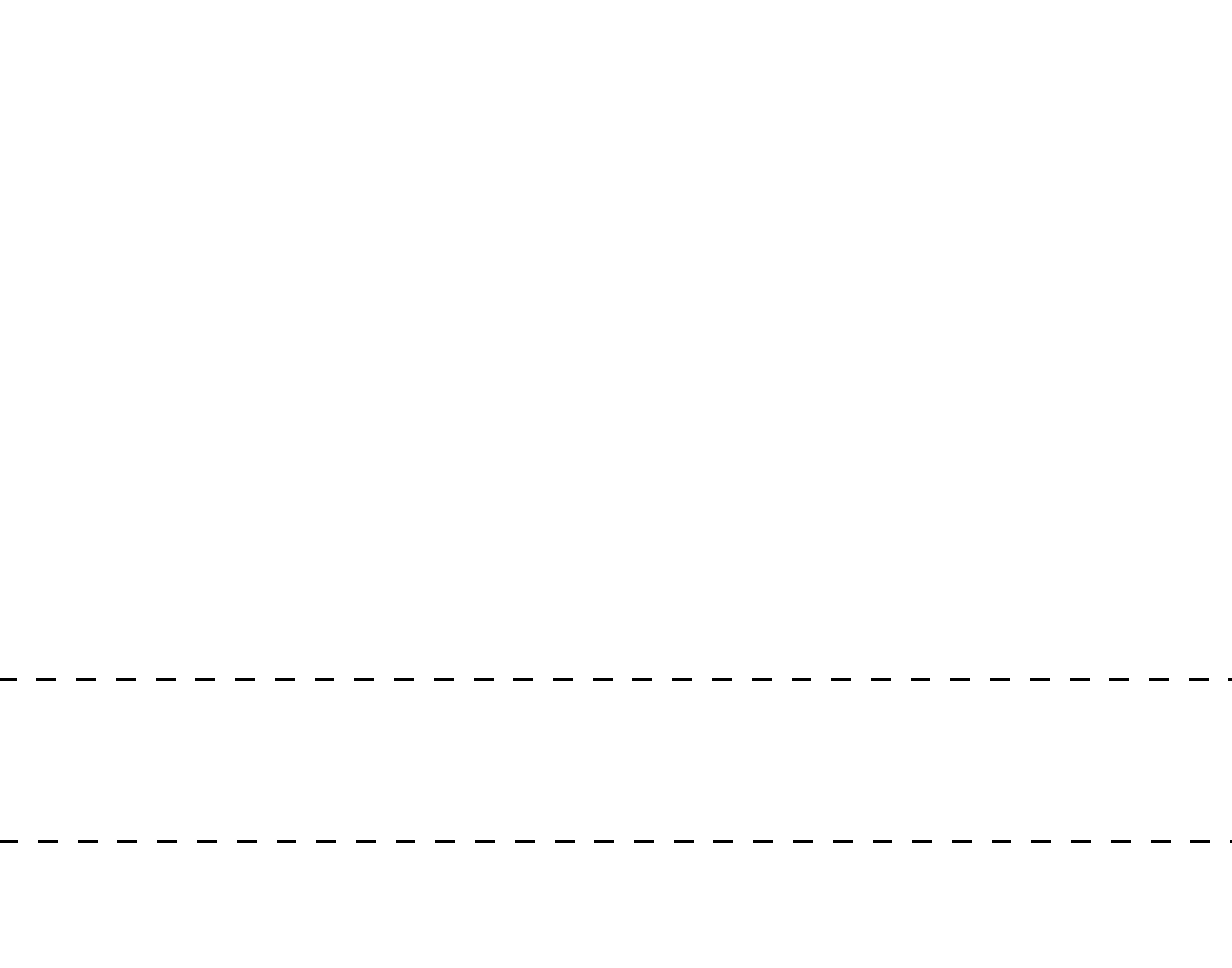
}
\hfill
\subcaptionbox[Short Subcaption]{(Case 2) $r$ changes its light to \texttt{switch off} as $r^4_1$ and $r^4_2$ change their lights to \texttt{symmetry}.
       \label{}
}
[
    0.48\textwidth 
]
{
    \fontsize{8pt}{8pt}\selectfont
    \def\svgwidth{0.48\textwidth}
    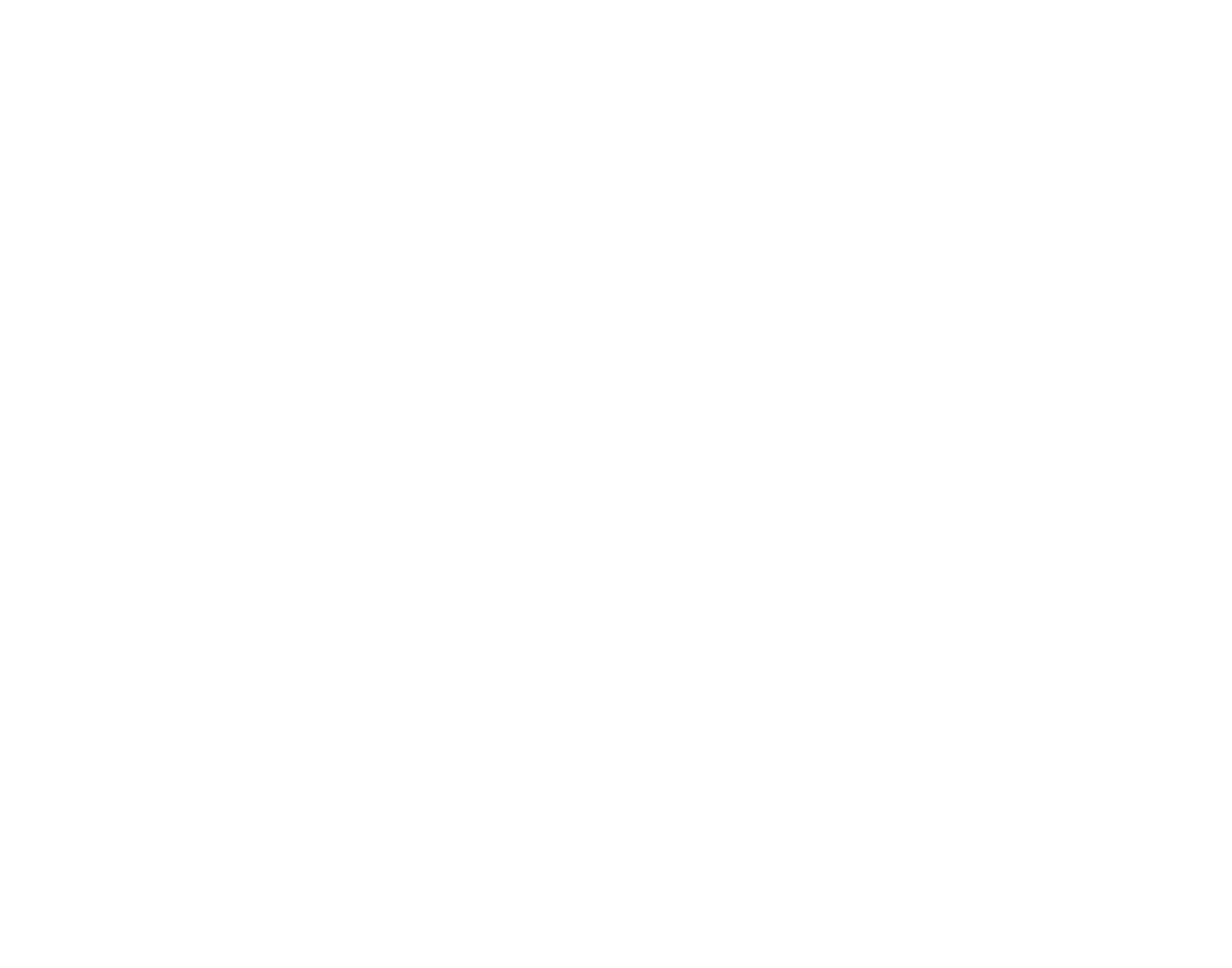
}
\hfill
\subcaptionbox[Short Subcaption]{
     (Case 2) $r^4_1$ and $r^4_2$ change their lights to \texttt{off}. Then $r$ will move horizontally left and then vertically.
}
[
    0.48\textwidth 
]
{
    \fontsize{8pt}{8pt}\selectfont
    \def\svgwidth{0.48\textwidth}
    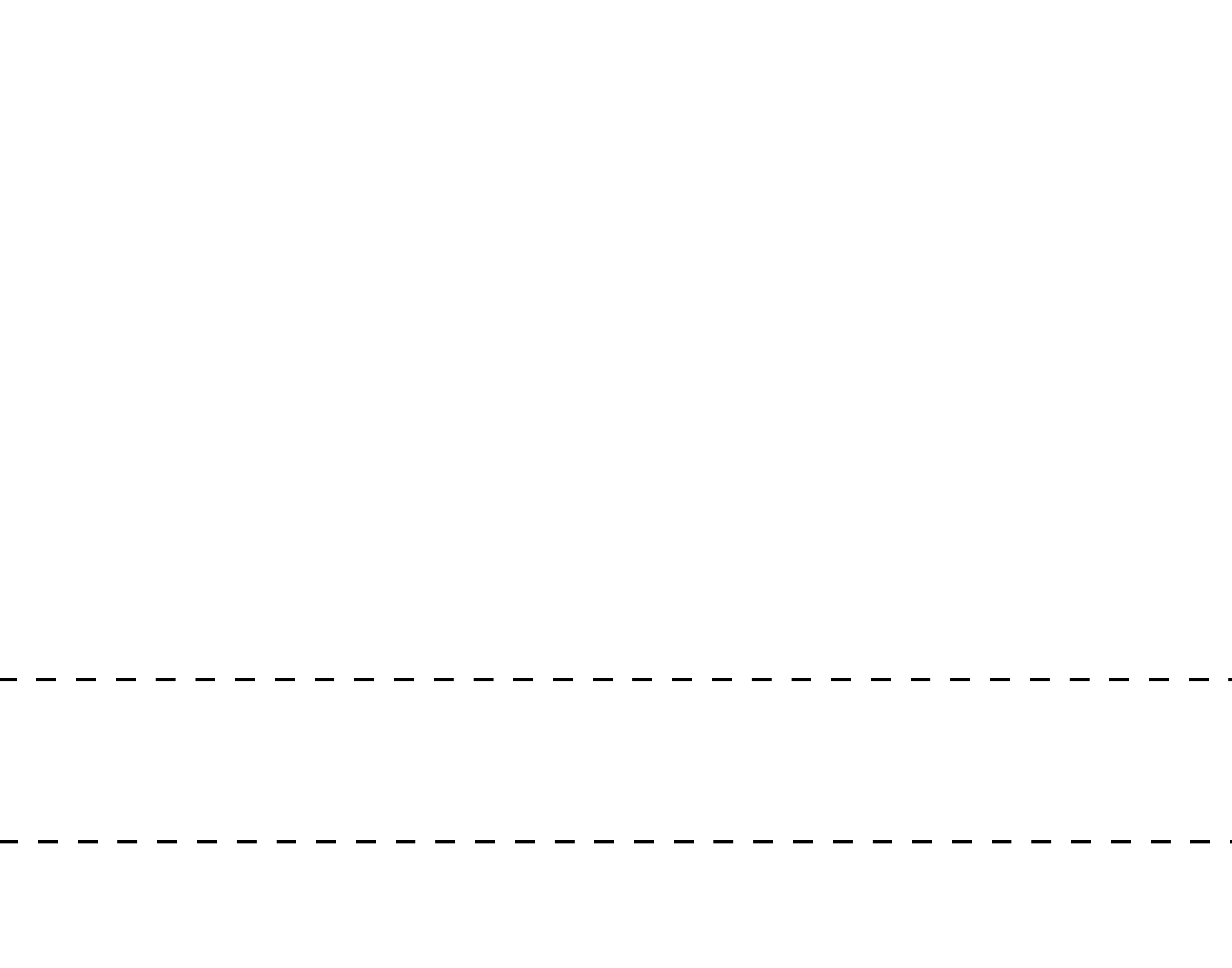
}
\hfill

\subcaptionbox[Short Subcaption]{ When $r$ aligns itself with $B_3$ all robots of $B_2$ and $B_4$ can see it. So, any robot of $B_2$ and $B_4$ with light not set to \texttt{off} will change its light to \texttt{off}. 
     \label{}
}
[
    0.48\textwidth 
]
{
    \fontsize{8pt}{8pt}\selectfont
    \def\svgwidth{0.48\textwidth}
    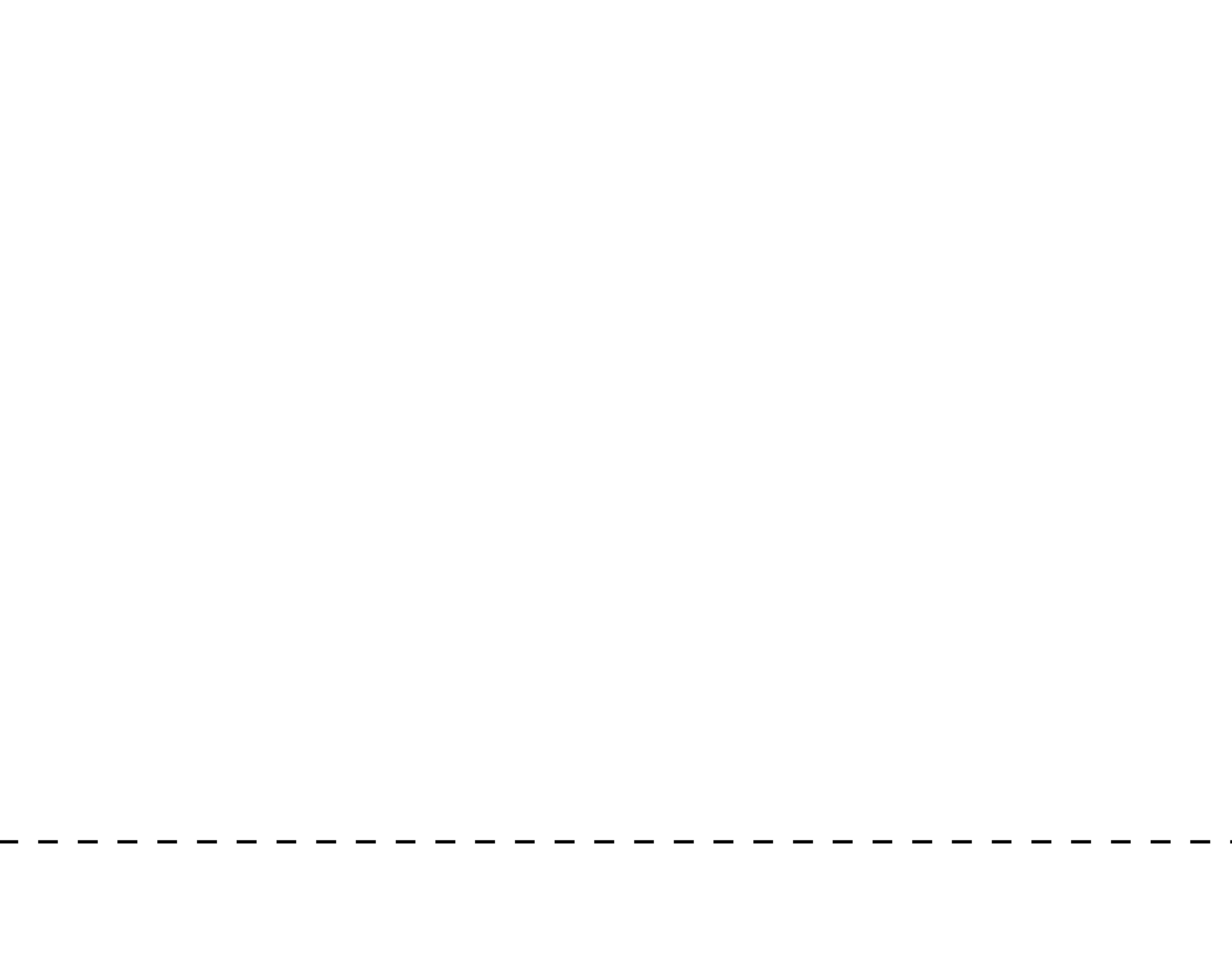
}
\hfill
\subcaptionbox[Short Subcaption]{ When $r$ aligns itself with $B_1$, $r^1_1$ thinks that it is in Phase 1, and changes its light to \texttt{terminal}. However, when $r$ moves further left, $r^1_1$ will change its light back to  \texttt{off}.
     \label{}
}
[
    0.48\textwidth 
]
{
    \fontsize{8pt}{8pt}\selectfont
    \def\svgwidth{0.48\textwidth}
    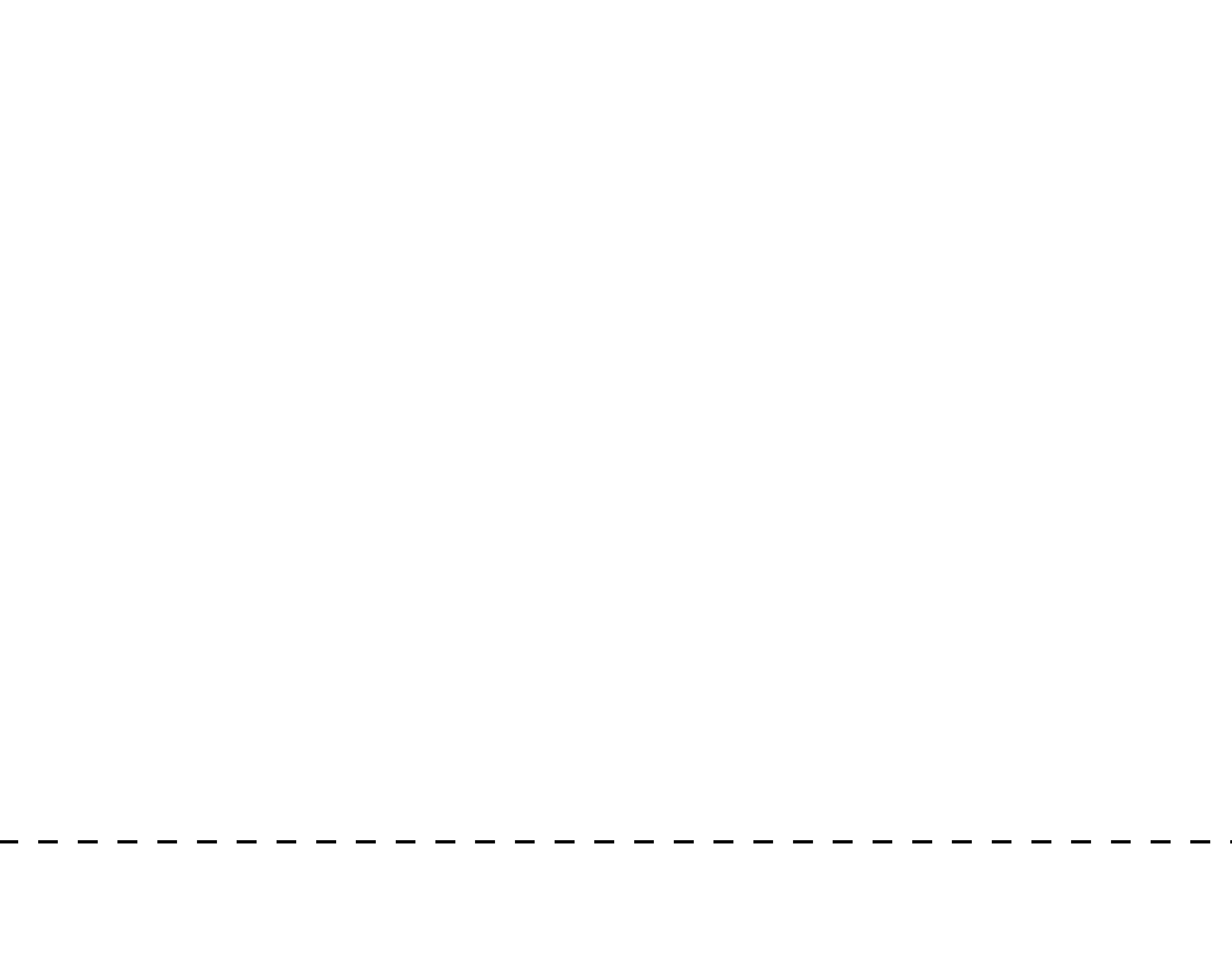
}
\caption[Short Caption]{The execution of \textsc{BecomeLeader()} in Phase 2.}
\label{Fig: become2}
\end{figure}

 \clearpage

  \section{Correctness of Stage 2}\label{appendix_s2}

  \begin{theorem}\label{thm_stage2}
  If $\mathbb{C}(T_1)$ is a leader configuration, then $\exists$ $T_2 > T_1$ such that $\mathbb{C}(T_2)$ is a final configuration similar to the given pattern.
\end{theorem}
  
  \begin{proof}
  Ordering of the robots and the target points are shown in Fig. \ref{fig: stage2_robot_order} and Fig. \ref{fig: stage2_robot_target} respectively. We shall first show that the non-leader robots $r_1, \ldots, r_{n-1}$  will move sequentially and place themselves on $\mathcal{L}_H(r_l)$ so that at some $T'' > T_1$, we shall have $r_0 = r_l, r_1, \ldots, r_{n-1}$ at $(0,-2), (1,-2), \ldots, (n-1, -2)$ (See Fig. \ref{fig: stage2_robot_line}). At the beginning, $r_1$ will find that the conditions of lines \ref{code s2: 0} and \ref{code s2: 1} (in Algorithm \ref{main_algorithm: stage 2}) are satisfied, while that of line \ref{code s2: 2} does not hold. So, it will decide to move to $(1,-2)$. Now suppose that $r_0 = r_l, r_1, \ldots, r_{i}$ $(i \geq 1)$ are at $(0,-2), (1,-2), \ldots, (i, -2)$, then $r_{i+1}$ will find that the condition of line \ref{code s2: 3} is satisfied and therefore will decide to move to $(i+1,-2)$. 
   
   Notice that when any robot $r_{j}, 1 \leq j \leq n-1,$ is moving, even if its move is interrupted due to the non-rigid movement assumption, it will again decide to move to $(j,-2)$ in the next cycle. Also, during the move, no other robot will decide to move. This is because 1) (if $j < n-1$) for $r_{j+1}, \ldots, r_{n-1}$, the condition of line \ref{code s2: 0} does not hold, 2) (if $j > 1$) for $r_{1}, \ldots, r_{j-1}$, the condition of line \ref{code s2: 4} does not hold and 3) for $r_l$, the condition of line \ref{code s2: 5} does not hold. Also notice  that since the robots have a physical extent, depending on the starting position and destination, a robot may not be able to move to its destination linearly in one go. In that case, it will move in a piecewise linear path (with at most three segments in case of rigid movement) as shown in Fig. \ref{fig: stage2_robot_move}. In that case, the `Move to $(j,-2)$' instructions in the pseudocode should be understood accordingly. Therefore, at some $T'' > T_1$, we shall have a stable configuration with $r_0 = r_l, r_1, \ldots, r_{n-1}$ at $(0,-2), (1,-2), \ldots, (n-1, -2)$. 
   
   Now we shall show that there is $T_2 > T''$, such that $\mathbb{C}(T_2)$ is a final configuration similar to the given pattern. Notice that in this configuration, the agreement in `up' and `down' is lost. When $r_1$ takes a snapshot for the first time after $T''$, it finds the conditions of line \ref{code s2: 4} to be true as both $\mathcal{H}_U^O(r_1)$ and $\mathcal{H}_B^O(r_1)$ (according to its own notion of `up' and `down') has no robots. Since $r_1$ can see $r_l$, it will take the center of $r_l$ as $(0,-2)$ and fix the positive direction of $Y$-axis according to it own notion of `up', compute the point $t_0$ in that coordinate system and move to that point. After the move, it will turn its light to \texttt{done}. Now suppose that $r_1, \ldots, r_{i}$ $(i \geq 1)$ are at $t_0, \ldots, t_{i-1}$ with light set to \texttt{done}. Obviously, here $r_{i+1}$ will have no ambiguity regarding the $Y$-axis. So, $r_{i+1}$ will find the conditions of line \ref{code s2: 4} to be true, compute the point $t_{i}$ in the agreed coordinate system and move accordingly.
   
   As before, the movements of the robots may not be linear due to their physical extent (See Fig. \ref{fig: stage2_robot_move2}). The move of any robot $r_{j}, 1 \leq j \leq n-1,$ may be interrupted due to the non-rigid movement assumption. The ordering of the target points are such that when it takes snapshot in the next cycle, it will be able to see all the robots on $\mathcal{L}_H(r_l)$ and hence, will find the condition of line \ref{code s2: 6} holding (with $i = n-j-1$). So, it will recompute $t_{j-1}$ and move accordingly. Now we shall argue that no other robot will decide to move during the move of any robot $r_{j}, 1 \leq j \leq n-1$. This is because during the movement of any robot $r_{j}$, all the robots on $\mathcal{L}_H(r_l)$ will be able to see it with light set to \texttt{off}. Therefore, 1) (if $j < n-1$) for $r_{j+1}, \ldots, r_{n-1}$, the condition of line \ref{code s2: 4} does not hold, and 2) for $r_l$, the condition of line \ref{code s2: 5} does not hold. Also, (if $j > 1$) for $r_{1}, \ldots, r_{j-1}$, their lights are set to \texttt{done} and hence, they will remain stable.  Now, only after all of $r_1, \ldots, r_{n-1}$ have completed their moves and turned their lights to \texttt{done}, $r_l = r_0$ will find condition of line \ref{code s2: 5} to hold. Then $r_0$ will move towards $t_{n-1}$. If it stops in between, it can identify $r_{n-1}$ from its local view. Therefore, it knows the points $t_{n-2}$ on the plane and knows its coordinate in the agreed coordinate system from the given input. From this it can recompute $t_{n-1}$ and hence will eventually reach there. Therefore, the given pattern is formed at some $T_2 > T_1$. \qed
  \end{proof}

 \begin{figure}[p]
\centering
\subcaptionbox[Short Subcaption]{ A leader configuration with the leader $r_l$ with light \texttt{leader} at $(0,-2)$ in the agreed coordinate system.
       \label{fig: stage2_robot_order}
}
[
    0.48\textwidth 
]
{
    \fontsize{8pt}{8pt}\selectfont
    \def\svgwidth{0.48\textwidth}
    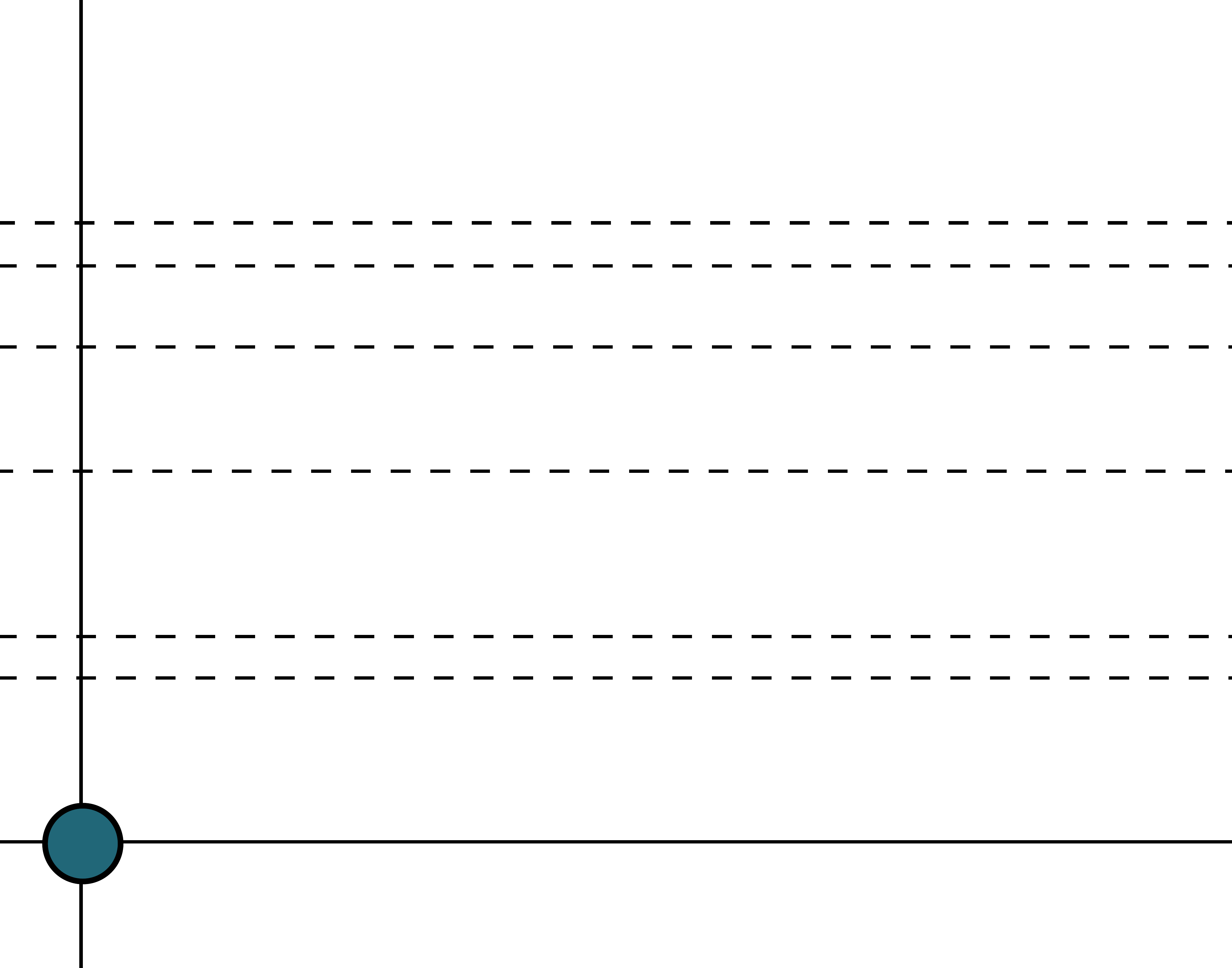
}
\hfill
\subcaptionbox[Short Subcaption]{ The pattern $\mathbb{P}$ embedded in the agreed coordinate system.
     \label{fig: stage2_robot_target}
}
[
    0.48\textwidth 
]
{
    \fontsize{8pt}{8pt}\selectfont
    \def\svgwidth{0.48\textwidth}
    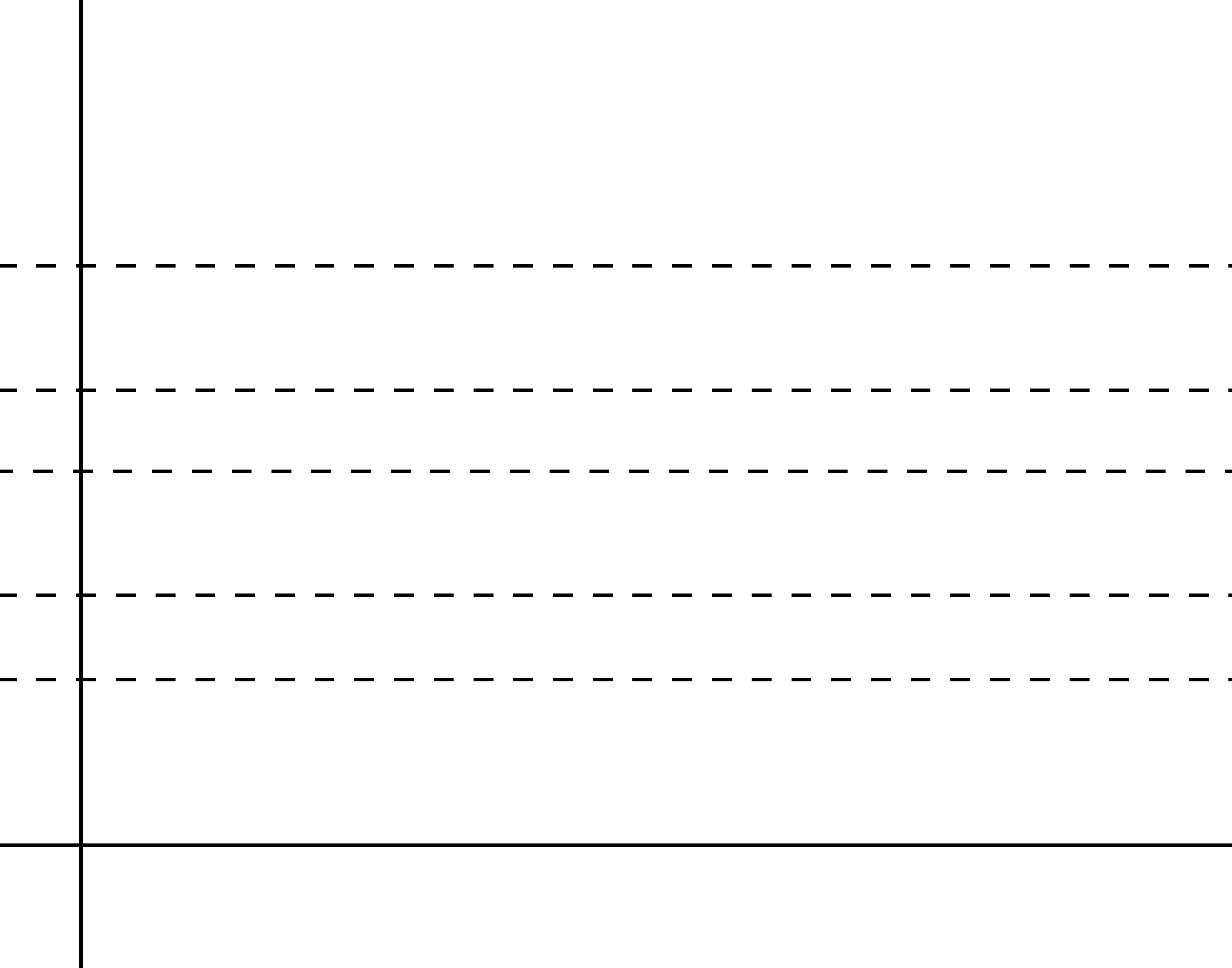
}
\hfill
\subcaptionbox[Short Subcaption]{ Movement of $r_8$ to $\mathcal{L}_H(r_l)$.
     \label{fig: stage2_robot_move}
}
[
    0.48\textwidth 
]
{
    \fontsize{8pt}{8pt}\selectfont
    \def\svgwidth{0.48\textwidth}
    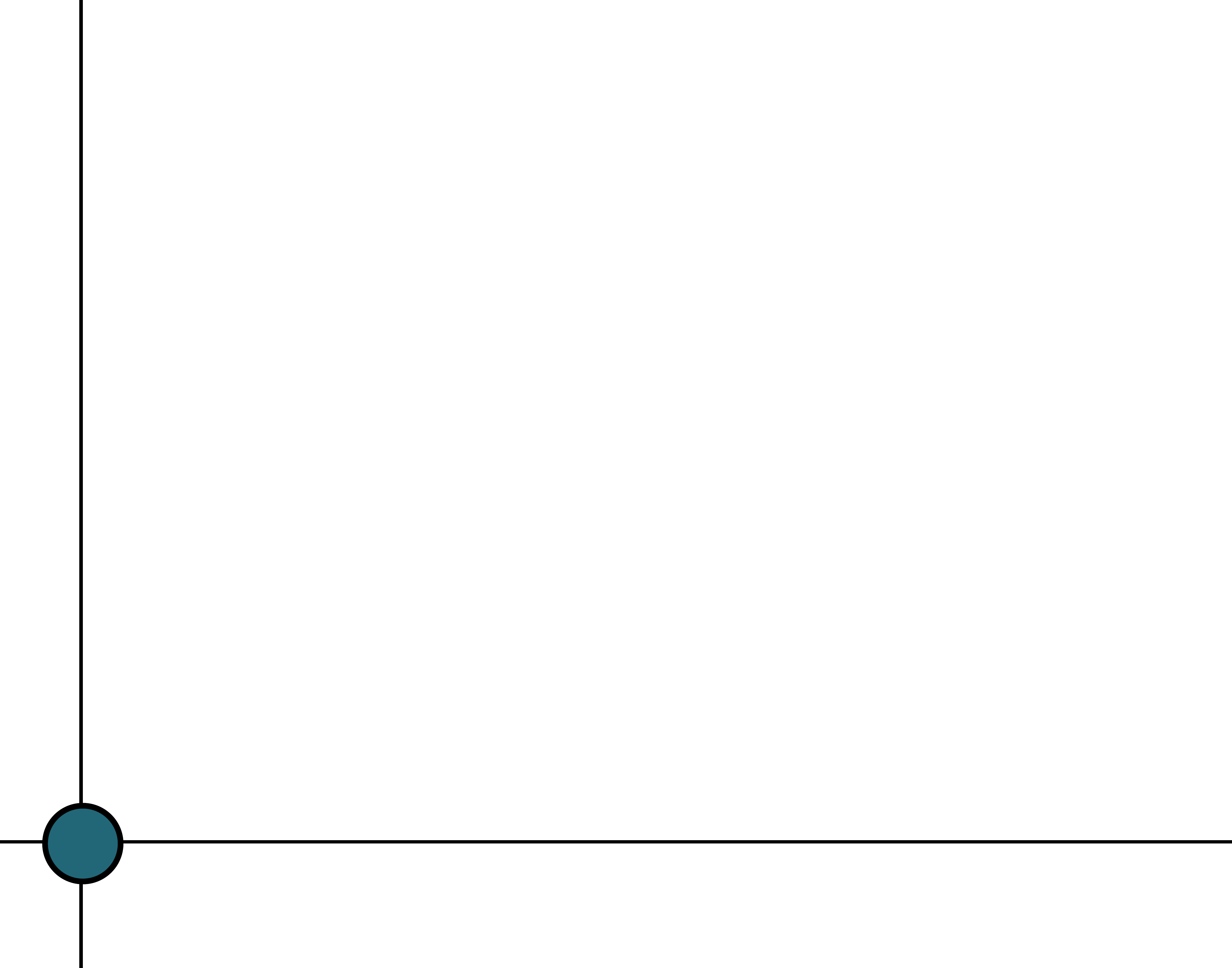
}
\hfill
\subcaptionbox[Short Subcaption]{ $r_0 = r_l, r_1,$ $\ldots, r_{n-1}$ are at $(0,-2), (1,-2),$ $\ldots, (n-1, -2)$.
     \label{fig: stage2_robot_line}
}
[
    0.48\textwidth 
]
{
    \fontsize{8pt}{8pt}\selectfont
    \def\svgwidth{0.48\textwidth}
    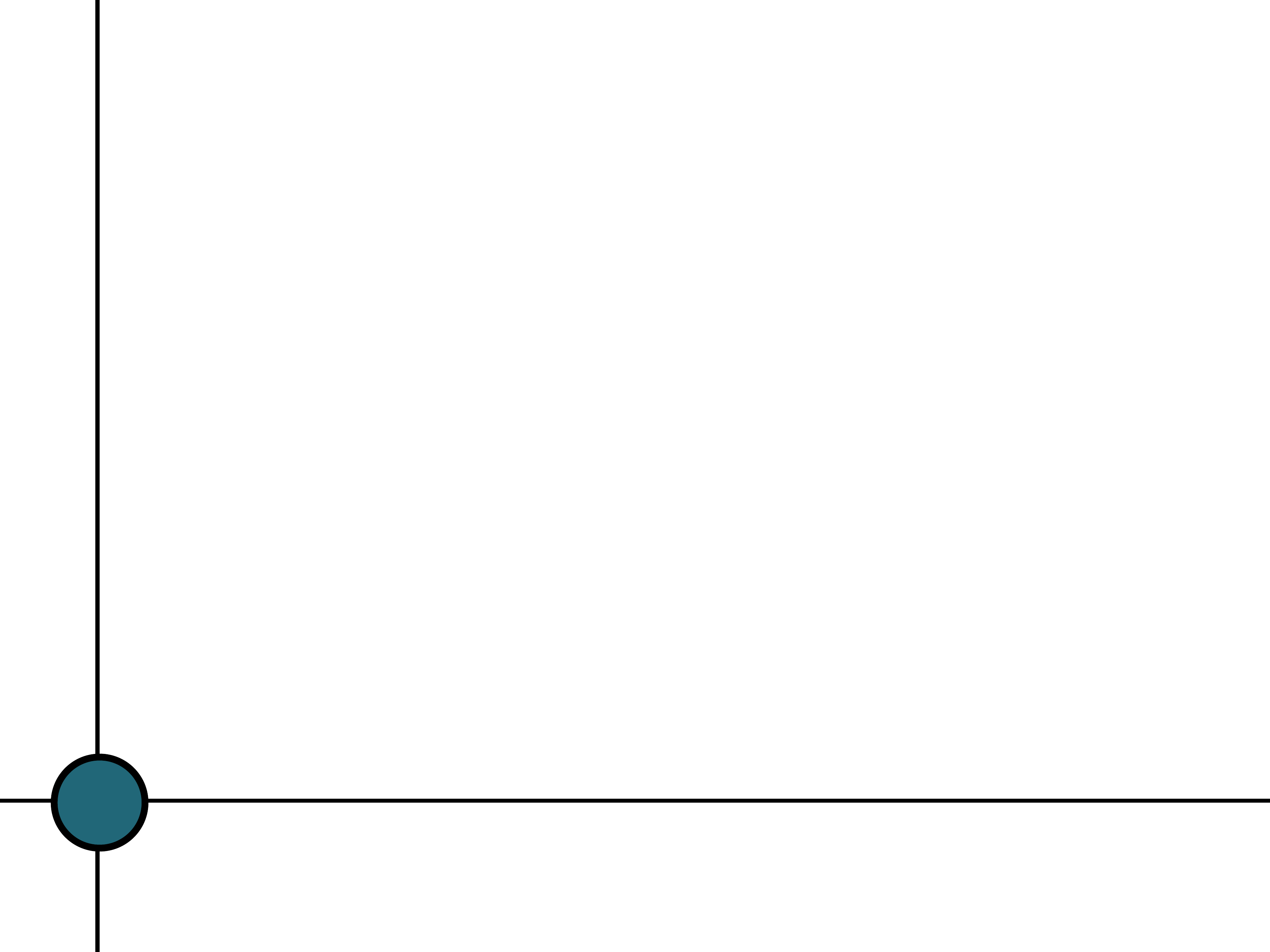
}
\hfill
\subcaptionbox[Short Subcaption]{ Movement of $r_3$ to $t_2$.
     \label{fig: stage2_robot_move2}
}
[
    0.48\textwidth 
]
{
    \fontsize{8pt}{8pt}\selectfont
    \def\svgwidth{0.48\textwidth}
    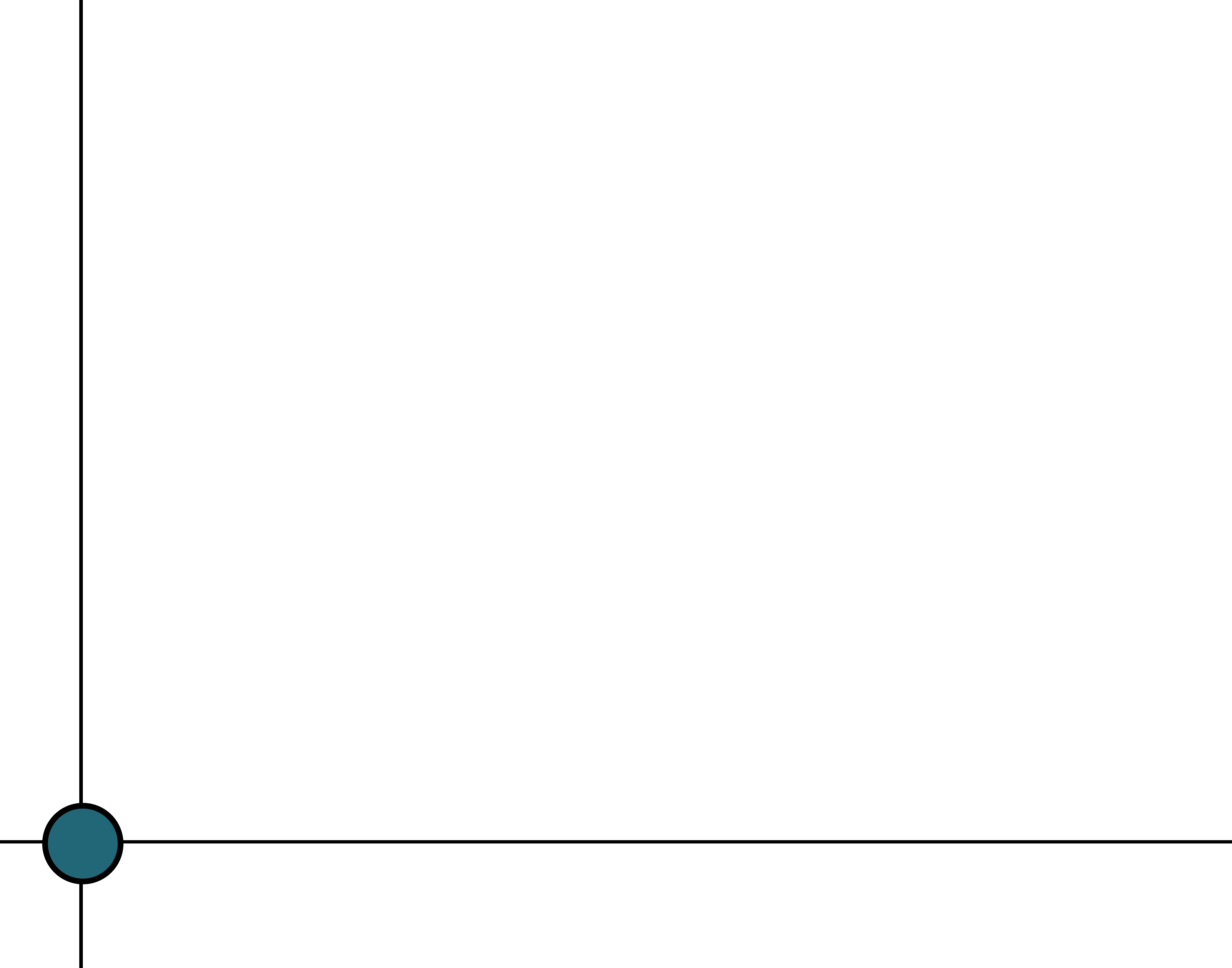
}
\hfill
\subcaptionbox[Short Subcaption]{Movement of $r_0$ to $t_{10}$.
     \label{fig: stage2_final}
}
[
    0.48\textwidth 
]
{
    \fontsize{8pt}{8pt}\selectfont
    \def\svgwidth{0.48\textwidth}
    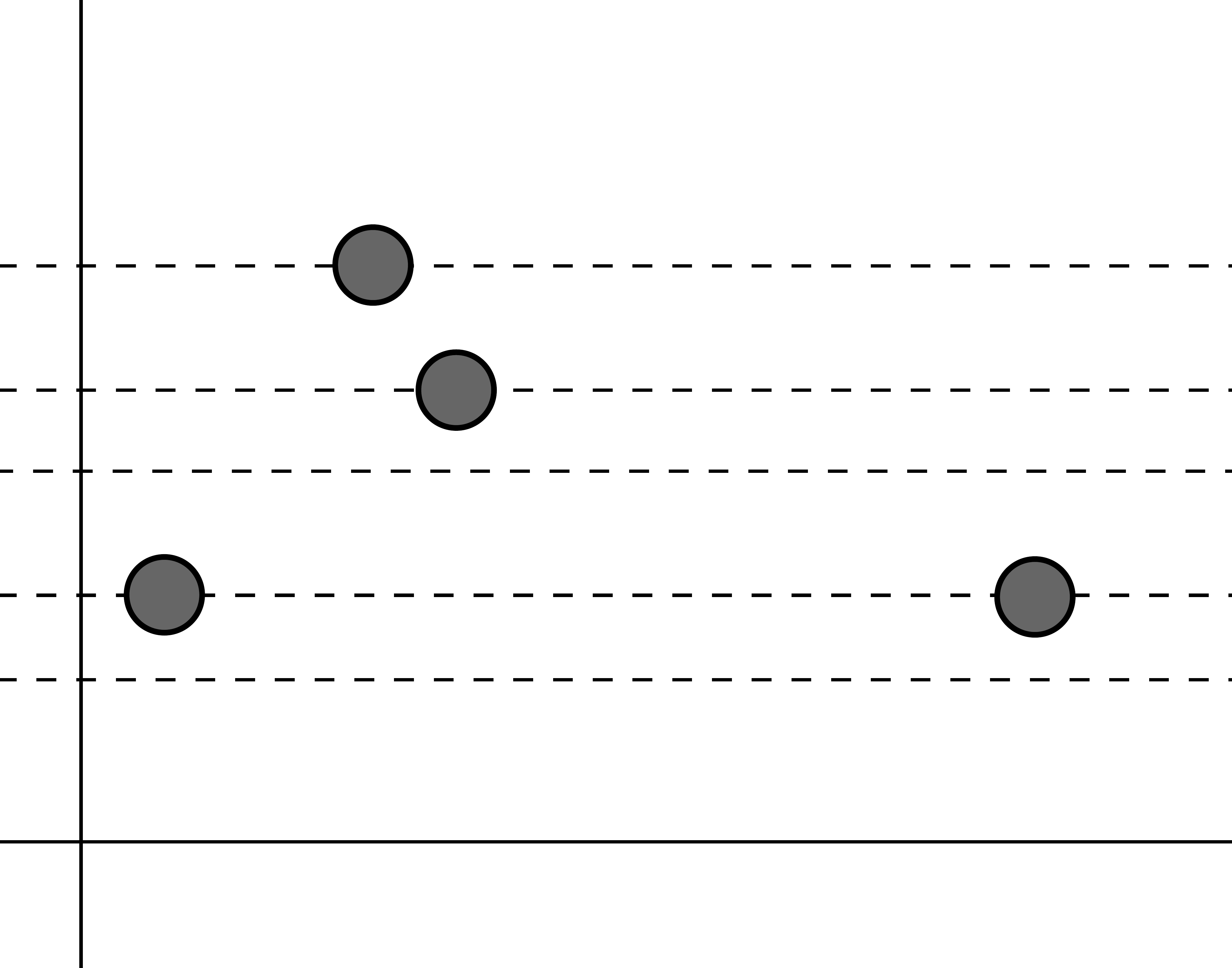
}
\caption[Short Caption]{Execution of Stage 2. }
\label{}
\end{figure}

\clearpage

\section{Efficiency of the Algorithm}

We shall study the efficiency of the algorithm in terms of the total number of moves executed by all the robots in the team. For the efficiency analysis, we shall consider rigid movements and semi-synchronous scheduler. That is, we shall assume the following.
\begin{enumerate}
 \item Time is logically divided into global rounds. In each round, a finite but non-zero number of the robots are activated. Every robot is activated infinitely often.
 
 \item In each round, all activated robots take the snapshots at the same time, and then perform their moves simultaneously, and completes their moves before the end of the round. As a result, no robot $r$ sees another robot $r'$ while $r'$ is moving.

 \item Each robot is able to reach its computed destination without any interruption.
 
 \end{enumerate}

 In this setting, we shall calculate the total number of moves required, in the worst case, in order to reach the final configuration. In Stage 1, all robots, except the one that eventually becomes the leader, execute $O(1)$ moves each. The robot that becomes the leader, will need $\Theta(n)$ moves, in the worst-case, while executing \textsc{BecomeLeader()}. In Stage 2, it is easy to see that all the robots need $O(1)$ moves each. Therefore, the total number of moves executed by all the robots in the team is $\Theta(n)$. This is also asymptotically optimal. To see this, consider an initial configuration where all robots are collinear. Then if the pattern to be formed does not have three collinear robots, then at least $n-2$ robots need to move. So, the total number of moves required to solve $\mathcal{APF}$ is $\Omega(n)$. Therefore, we can conclude as the following.

\begin{theorem}\label{thm_stage2}
 Algorithm \ref{main_algorithm} solves $\mathcal{APF}$ in asymptotically optimal number of moves.  
\end{theorem}

\end{document}